\documentclass[11pt]{article}
\usepackage[latin1]{inputenc}
\usepackage{fullpage}  
\usepackage{graphicx}  

\usepackage{amsmath}   
\usepackage{amsfonts}  
\usepackage{amssymb}   
\usepackage{amsthm}    
\usepackage{mathrsfs}

\newtheorem{theorem}{Theorem}

\newtheorem{corollary}{Corollary}
\newtheorem{lemma}{Lemma}

\newtheorem{definition}{Definition}

\newcommand{\BO}[1]{{O}\mathopen{}\left(#1\right)\mathclose{}}

\newcommand{\BT}[1]{{\Theta}\mathopen{}\left(#1\right)\mathclose{}}
\newcommand{\BOM}[1]{{\Omega}\mathopen{}\left(#1\right)\mathclose{}}

\newcommand{\A}{\mathcal{A}}
\newcommand{\B}{\mathcal{B}}
\newcommand{\C}{\mathcal{C}}
\newcommand{\bfm}[1]{\mbox{\boldmath $#1$}}
\newcommand{\of}{\sigma}
\newcommand{\spec}[1]{$M(#1)$}

\newcommand{\eval}[1]{$M(#1)$}
\newcommand{\evald}{\eval{p,\of}}
\newcommand{\exec}[1]{D-BSP$(#1)$}
\newcommand{\execd}{\exec{p,\bfm{g},\bfm{\ell}}}
\newcommand{\sync}[1][i]{\texttt{sync$(#1)$}}
\newcommand{\send}{\texttt{send$(m,q)$}}
\newcommand{\receive}{\texttt{receive$()$}}
\newcommand{\GAP}{\mbox{\rm GAP}}

\title{Network-Oblivious Algorithms\thanks{This work was
  supported, in part, by MIUR of Italy under project AMANDA and by the University of Padova
  under projects STPD08JA32 and CPDA121378/12. A preliminary version of this paper appeared
  in \textit{Proceedings of the 21st IEEE International Parallel and Distributed Processing
  Symposium (IPDPS)}, 2007.}}
	
\author{Gianfranco Bilardi\thanks{Department of Information Engineering, University of
  Padova, 35131 Padova, Italy. \hbox{E-mail}:~\texttt{bilardi@dei.unipd.it}.}
  \and
  Andrea Pietracaprina\thanks{Department of Information Engineering, University of
  Padova, 35131 Padova, Italy. \hbox{E-mail}:~\texttt{capri@dei.unipd.it}.}
  \and
  Geppino Pucci\thanks{Department of Information Engineering, University of
  Padova, 35131 Padova, Italy. \hbox{E-mail}:~\texttt{geppo@dei.unipd.it}.}
  \and
  Michele Scquizzato\thanks{Department of Computer Science, University of
  Pittsburgh, Pittsburgh, PA 15260, USA. \hbox{E-mail}:~\texttt{scquizza@pitt.edu}.
  Most of this work was done while this author was a Ph.D.\ student at the
  University of Padova.}
  \and
  Francesco Silvestri\thanks{Department of Information Engineering, University of
  Padova, 35131 Padova, Italy. \hbox{E-mail}:~\texttt{silvest1@dei.unipd.it}.}}

\date{}

\begin{document}

\maketitle

\begin{abstract}
A framework is proposed for the design and analysis of \emph{network-oblivious algorithms},
namely, algorithms that can run unchanged, yet efficiently, on a variety of machines
characterized by different degrees of parallelism and communication capabilities.
The framework prescribes that a network-oblivious algorithm be specified on a parallel
model of computation where the only parameter is the problem's input size, and then
evaluated on a model with two parameters, capturing parallelism granularity and
communication latency. It is shown that, for a wide class of network-oblivious algorithms,
optimality in the latter model implies optimality in the Decomposable BSP model,
which is known to effectively describe a wide and significant class of parallel
platforms. The proposed framework can be regarded as an attempt to port the notion of
obliviousness, well established in the context of cache hierarchies, to the realm
of parallel computation. Its effectiveness is illustrated by providing optimal
network-oblivious algorithms for a number of key problems. Some limitations of
the oblivious approach are also discussed.
\end{abstract}

\newpage

\section{Introduction}\label{sec:introduction}

Communication plays a major role in determining the performance of
algorithms on current computing systems and has a
considerable impact on energy consumption. Since the relevance of 
communication increases with the size of the system, it is expected to
play an even greater role in the future. Motivated by this scenario,
a large body of results have been devised concerning the design and
analysis of communication-efficient algorithms.  While often useful
and deep, these results do not yet provide a coherent and unified
theory of the communication requirements of computations. One major
obstacle toward such a theory lies in the fact that, \emph{prima facie},
communication is defined only with respect to a specific mapping of a
computation onto a specific machine structure.  Furthermore, the
impact of communication on performance depends on the latency and
bandwidth properties of the channels connecting different parts of the
target machine. Hence, the design, optimization, and analysis of
algorithms can become highly machine-dependent, which is undesirable
from the economical perspective of developing efficient and portable
software. The outlined situation has been widely recognized, and a
number of approaches have been proposed to solve it or to mitigate it.

On one end of the spectrum, we have the \emph{parallel slackness}
approach, based on the assumption that, as long as a sufficient amount
of parallelism is exhibited, general and automatic
latency-hiding techniques can be deployed to achieve an efficient
execution.  Broadly speaking, the required algorithmic parallelism
should be proportional to the product of the number of processing units
by the worst-case latency of the target
machine~\cite{Valiant90}. Further assuming that this amount of
parallelism is available in computations of practical
interest, algorithm design can dispense altogether with communication
concerns and focus on the maximization of parallelism. The
functional/data-flow and the PRAM models of computations have often
been supported with similar arguments.  Unfortunately, as argued
in~\cite{BilardiP95,BilardiP97,BilardiP99}, latency hiding is not a
scalable technique, due to fundamental physical constraints.  Hence,
parallel slackness does not really solve the communication
problem. (Nevertheless, functional and PRAM models are quite valuable
and have significantly contributed to the understanding of other
dimensions of computing.)

On the other end of the spectrum, we could place the
\emph{universality} approach, whose objective is the development of
machines (nearly) as efficient as any other machine of (nearly) the
same cost, at executing any computation (see, e.g.,
\cite{Leiserson85,BilardiP95,BhattBP08,BilardiP11a}). To the extent
that a universal machine with very small performance and cost gaps
could be identified, one could adopt a model of computation
sufficiently descriptive of such a machine, and focus most of the
algorithmic effort on this model.  As technology approaches the
inherent physical limitations to information processing, storage, and
transfer, the emergence of a universal architecture becomes more
likely. Economy of scale can also be a force favoring convergence in
the space of commercial machines.  While this appears as a perspective
worthy of investigation, at the present stage, neither the known
theoretical results nor the trends of commercially available platforms
indicate an imminent strong convergence.

In the middle of the spectrum, a variety of computational models proposed
in the literature can be viewed as variants of an approach aiming at
realizing an \emph{efficiency/portability/design-complexity tradeoff}
\cite{BilardiP11}. Well-known examples of these models are
LPRAM~\cite{AggarwalCS90}, BSP~\cite{Valiant90} and its refinements
(such as D-BSP~\cite{delaTorreK96,BilardiPP07},
BSP*~\cite{BaumkerDH98}, E-BSP~\cite{JuurlinkW98}, and
BSPRAM~\cite{Tiskin98}), LogP~\cite{CullerKPSSSSE96},
QSM~\cite{GibbonsMR99}, and several others. These models aim at
capturing features common to most (reasonable) machines, while
ignoring features that differ.  The hope is that performance of real
machines be largely determined by the modeled features, so that
optimal algorithms in the proposed model translate into near optimal
ones on real machines. A drawback of these models is that they include
parameters that affect execution time. Then, in general, efficient
algorithms are parameter-aware, since different algorithmic strategies
can be more efficient for different values of the parameters.  One
parameter present in virtually all models is the number of processors.
Most models also exhibit parameters describing the time required to
route certain communication patterns. Increasing the number of
parameters, from just a small constant to logarithmically many in the
number of processors, can considerably increase the effectiveness of
the model with respect to realistic architectures, such as
point-to-point networks, as extensively discussed
in~\cite{BilardiPP07}. A price is paid in the increased complexity of
algorithm design necessary to gain greater efficiency across a larger
class of machines. The complications further compound if the
hierarchical nature of the memory is also taken into account, so that
communication between processors and memories becomes an optimization
target as well.

It is natural to wonder whether, at least for some problems, parallel
algorithms can be designed that, while independent of any
machine/model parameters, are nevertheless efficient for wide ranges
of these parameters.  In other words, we are interested in exploring
the world of efficient \emph{network-oblivious} algorithms with a
spirit similar to the one that motivated the development of efficient
\emph{cache-oblivious} algorithms~\cite{FrigoLPR12}.  In this paper,
we define the notion of network-oblivious algorithms and propose a
framework for their design and analysis.  Our framework is based on
three models of computation, each with a different role, as briefly
outlined below.

The three models are based on a common organization consisting of a
set of CPU/memory nodes communicating through some interconnection.
Inspired by the BSP model and its aforementioned variants, we assume
that the computation proceeds as a sequence of supersteps, where in a
superstep each node performs local computation and sends/receives
messages to/from other nodes, which will be consumed in the subsequent
superstep. Each message occupies a constant number of words.

The first model of our framework (specification model), is used
to specify network-oblivious algorithms. In this model, the number of
CPU/memory nodes, referred to as virtual processors, is a
function $v(n)$ of the input size and captures the amount of
parallelism exhibited by the algorithm.  The second model
(evaluation model) is the basis for analyzing the performance
of network-oblivious algorithms on different machines. It is
characterized by two parameters, independent of the input: the number
$p$ of CPU/memory nodes, simply referred to as processors in
this context, and a fixed latency/synchronization cost $\sigma$ per
superstep. The communication complexity of an algorithm is defined in
this model as a function of $p$ and $\sigma$.  Finally, the third
model (execution machine model) enriches the evaluation model,
by replacing parameter $\sigma$ with two independent parameter vectors
of size logarithmic in the number of processors, which represent,
respectively, the inverse of the bandwidth and the latency costs of
suitable nested subsets of processors. In this model, the
communication time of an algorithm is analyzed as a function of $p$ and
of the
two parameter vectors.  In fact, the execution machine model of our
framework coincides with the Decomposable Bulk Synchronous Parallel
(D-BSP) model~\cite{delaTorreK96,BilardiPP07}, which is known to
describe reasonably well the behavior of a large class of
point-to-point networks by capturing their hierarchical structure
\cite{BilardiPP99}.

A \emph{network-oblivious algorithm} is designed in the specification
model but can be run on the evaluation or execution machine models by
letting each processor of these models carry out the work of a
pre-specified set of virtual processors.  The main contribution of
this paper is an optimality theorem showing that for a wide and
interesting class of network-oblivious algorithms, which satisfy some
technical conditions and whose communication requirements depend only
on the input size and not on the specific input instance, optimality
in the evaluation model automatically translates into optimality in
the D-BSP model, for suitable ranges of the models' parameters. It is
this circumstance that motivates the introduction of the intermediate
evaluation model, which simplifies the analysis of network-oblivious
algorithms, while effectively bridging the performance analysis to the
more realistic D-BSP model.

In order to illustrate the potentiality of the framework we devise
network-oblivious algorithms for several fundamental problems, such as
matrix multiplication, Fast Fourier Transform, comparison-based
sorting, and a class of stencil computations. In all cases, except for
stencil computations, we show, through the optimality theorem, that
these algorithms are optimal when executed on the D-BSP, for wide
ranges of the parameters. Unfortunately, there exist problems for
which optimality on the D-BSP cannot be attained in a
network-oblivious fashion for wide ranges of parameters. 
We show that this is the case for the broadcast problem.

To help place our network-oblivious framework into perspective, it may
be useful to compare it with the well-established sequential
cache-oblivious framework~\cite{FrigoLPR12}. In the latter, the
specification model is the Random Access Machine; the evaluation model
is the Ideal Cache Model IC$(\mathcal{M},\mathcal{B})$, with only one
level of cache of size $\mathcal{M}$ and line length $\mathcal{B}$;
and the execution machine model is a machine with a hierarchy of
caches, each with its own size and line length. In the cache-oblivious
context, the simplification in the analysis arises from the fact that,
under certain conditions, optimality on IC$(\mathcal{M},\mathcal{B})$,
for all values of $\mathcal{M}$ and $\mathcal{B}$, translates into
optimality on multilevel hierarchies.

The notion of obliviousness in parallel settings has been recently
addressed by several papers. In a preliminary version of this
work~\cite{BilardiPPS07} (see also~\cite{Herley11}), we proposed a
framework similar to the one presented here, where messages are packed
in blocks whose fixed size is a parameter of the evaluation and execution
machine models.  While blocked communication may be preferable for
models where the memory and communication hierarchies are seamlessly
integrated (e.g., multicores), latency-based models like the one used
here are equivalent for that scenario and also capture the case when communication
is accomplished through a point-to-point network. Later, Chowdhury et
al.~\cite{ChowdhuryRSB13} introduced a multilevel hierarchical model for
multicores and the notion of \emph{multicore-oblivious algorithm} for this model.
A multicore-oblivious algorithm is specified with no mention of any
machine parameters, such as the number of cores, number of cache levels,
cache sizes and block lengths. Cole and Ramachandran~\cite{ColeR10,ColeR12a,ColeR12b}
presented parallel algorithms for a shared-memory multicore model featuring
a two-level memory, which are oblivious to the number of processors and to
the memory parameters. Finally, Blelloch et al.~\cite{BlellochFGV11} introduced
a parallel version of the cache-oblivious framework in~\cite{FrigoLPR12},
named the Parallel Cache-Oblivious model, and described a scheduler for
oblivious irregular computations. In contrast to these approaches,
Valiant~\cite{Valiant11} studies parallel algorithms for multicore architectures
advocating a parameter-aware design of portable algorithms. In his paper,
he presents optimal algorithms for the Multi-BSP, a bridging model for
multicore architectures which exhibits a hierarchical structure akin
to that of our execution machine model. In fact, for the latter model
we are able to show that, for the same computational problems, optimal
parallel algorithms do exist that are parameter-free.

The rest of the paper is organized as follows. In
Section~\ref{sec:framework} we formally define the three models
relevant to the framework and in Section~\ref{sec:optimality} we prove
the optimality theorem mentioned above. In
Section~\ref{sec:algorithms}, we present the network-oblivious
algorithms for matrix multiplication, Fast Fourier Transform,
comparison-based sorting, and stencil computations. We also discuss
the impossibility result regarding the broadcast
problem. Section~\ref{sec:extension} extends the optimality theorem by
presenting a less powerful version which, however, applies to a wider
class of algorithms.  Section~\ref{sec:conclusions}
concludes the paper with some final remarks.

\section{The Framework}\label{sec:framework}
We begin by introducing a parallel machine model \spec{v}, which
underlies the specification, the evaluation, and the execution
components of our framework. Specifically, \spec{v} consists of a set
of $v$ processing elements, denoted by
$\text{P}_0,\text{P}_1,\dots,\text{P}_{v-1}$, each equipped with a CPU
and an unbounded local memory, which communicate through some
interconnection. (For simplicity, throughout this paper, $v$ is
assumed to be a power of two.)  The instruction set of each CPU is
essentially that of a standard Random Access Machine, augmented with
the three primitives \sync, \send, and \receive.  Furthermore, each
$\text{P}_r$ has access to its own index $r$ and to the number $v$ of
processing elements.  When $\text{P}_r$ invokes primitive \sync, with
$i$ in the integer range $[0,\log v)$, then it waits that each
processing element whose index shares with $r$ the $i$ most significant 
bits also perform a \sync, or
terminate its execution.\footnote{For notational convenience,
  throughout this paper we use $\log x$ to mean $\max\{1,\log_2
  x\}$.} In other words, \sync\ is a barrier which synchronizes the $v/2^i$ 
processing elements whose indices share the $i$ most significant bits.
When $\text{P}_r$ invokes \send, with $0 \leq q < v$, a constant-size
message $m$ is sent to $\text{P}_q$; the message will be available in
$\text{P}_q$ only after a \sync[k], where $k$ is not bigger than the
number of most significant bits shared by $r$ and $q$. On the other
hand, the function \receive\ returns an element in the set of messages
received up to the preceding barrier and removes it from the set.

In this paper, we restrict our attention to algorithms where the
sequence of labels of the sync operations is the same for all
processing elements, and where the last operation executed by each
processing element is a \texttt{sync}.\footnote{As we will see in the
  paper, several algorithms naturally comply or can easily be adapted
  to comply with these restrictions. Nevertheless, a less restrictive
  family of algorithms for \spec{v} can be defined, by allowing
  processing elements to feature different traces of labels of their
  sync operations, still ensuring termination.  The exploration of the
  potentialities of these algorithms is left for future research.}
In this case, the execution of an algorithm can be viewed as a
sequence of \emph{supersteps}, where a superstep consists of all 
operations performed between two consecutive sync operations,
including the second of these sync operations. Supersteps are labeled
by the index of their terminating sync operation: namely, a superstep
terminating with \sync\ will be referred to as an
\emph{$i$-superstep}, for $0 \leq i < \log v$.  Furthermore, we make
the reasonable assumption that in an $i$-superstep, each $\text{P}_r$
can send messages only to processing elements whose index agrees with
$r$ in the $i$ most significant bits, that is, message exchange occurs
only between processors belonging to the same synchronization subset.
We observe that the results of this paper would hold even if, in the
various models considered, synchronizations were not explicitly
labeled.  However, explicit labels can help reduce synchronization
costs. For instance, they become crucial for the efficient execution
of the algorithms on point-to-point networks, especially those of
large diameter.

Consider an \spec{v}-algorithm $\A$ satisfying the above restrictions.
For a given input instance $I$, we use $L^i_{\A}(I)$ to
denote the set of $i$-supersteps executed by $\A$ on input $I$, and
define $S^i_{\A}(I) = |L^i_{\A}(I)|$, for $0 \leq i < \log
v$. Algorithm $\A$ can be naturally and automatically adapted to
execute on a smaller machine \spec{p}, with $p < v$, by stipulating
that processing element $\text{P}_j$ of \spec{p} will carry out the
operations of the $v/p$ consecutively numbered processing elements of
\spec{v} starting with $\text{P}_{j(v/p)}$, for each $0\leq j<p$. We
call this adaptation \emph{folding}. Under folding, supersteps with a
label $i < \log p$ on \spec{v} become supersteps with the same label
on \spec{p}, while supersteps with label $i \geq \log p$ on \spec{v}
become local computation on \spec{p}. Hence, when considering the
communication occurring in the execution of $\A$ on \spec{p}, the set
$L^i_{\A}(I)$ is relevant as long as $i < \log p$. 

A \emph{network-oblivious algorithm} $\A$ for a given computational
problem $\Pi$ is designed on \spec{v(n)}, referred to as
\emph{specification model}, where the number $v(n)$ of processing elements,
which is a function of the input size, is selected as part of the
algorithm design. The processing elements are called \emph{virtual
processors} and are denoted by
$\text{VP}_0,\text{VP}_1,\dots,\text{VP}_{v(n)-1}$, in order to
distinguish them from the processing elements of the other two
models. Since the folding mechanism illustrated above enables $\A$ to
be executed on a smaller machine \spec{p}, the design effort can be
kept focussed on just one convenient virtual machine size, oblivious
to the actual number of processors on which the algorithm will be
executed. 

While a network-oblivious algorithm is specified for a large virtual
machine, it is useful to analyze its communication requirements on
machines with reduced degrees of parallelism.  For these purposes, we
introduce the \emph{evaluation model} \eval{p,\of}, where $p\geq 1$ is
power of two and $\of \geq 0$, which is essentially an \spec{p} where
the additional parameter $\of$ is used to account for the latency plus
synchronization cost of each superstep. The processing elements of
\eval{p,\of} are called \emph{processors} and are denoted by
$\text{P}_0,\text{P}_1,\dots,\text{P}_{p-1}$.  Consider the execution
of an algorithm $\A$ on \eval{p, \of} for a given input $I$.  For each
superstep $s$, the metric of interest that we use to evaluate the
communication requirements of the algorithm is the maximum number of
messages $h^s_{\A}(I,p)$ sent/destined by/to any processor in that
superstep. Thus, the set of messages exchanged in the superstep can be
viewed as forming an $h^s_{\A}(I,p)$-\emph {relation}, where
$h^s_{\A}(I,p)$ is often referred to as the \emph {degree} of the
relation. In the evaluation model, the communication cost of a
superstep of degree $h$ is defined as $h + \of$, and it is independent
of the superstep's label.  For our purposes, it is convenient to
consider the cumulative degree of all $i$-supersteps, for $0 \leq i <
\log p$:
\begin{eqnarray*}
F^i_{\A}(I,p) = \sum_{s \in L^i_{\A}(I)} h^s_{\A}(I,p).
\end{eqnarray*}
Then, the \emph {communication complexity} of $\A$ on \eval{p,\of} is
defined as
\begin{equation}\label{eq:commcompNO}
H_{\A}(n,p,\of) =
\max_{I:|I|=n} \left\{\sum_{i=0}^{\log p -1}\left( F^i_{\A}(I,p) + 
S^i_{\A}(I) \cdot \of\right)\right\}.
\end{equation}
\sloppy
We observe that the evaluation model with this performance metric
coincides with the BSP model~\cite{Valiant90} where the
bandwidth parameter $g$ is set to 1 and the latency/synchronization
parameter $\ell$ is set to $\of$.

Next, we turn our attention to the last model used in the framework,
called \emph{execution machine model}, which represents the machines
where network-oblivious algorithms are actually executed.
We focus on parallel machines whose underlying interconnection
exhibits a hierarchical structure, and use the \emph{Decomposable BSP}
(\emph{D-BSP}) model~\cite{delaTorreK96,BilardiPP07} as our execution
machine model. A \execd\/, with $\bfm{g} = (g_0, \dots, g_{\log p-1})$
and $\bfm{\ell} = (\ell_0, \dots, \ell_{\log p-1})$, is an
\spec{p} where the cost of an $i$-superstep depends on parameters $g_i$
and $\ell_i$, for $0\leq i < \log p$. The processing elements, called
processors and denoted by $\text{P}_0,\text{P}_1,\dots,\text{P}_{p-1}$
as in the evaluation model, are partitioned into nested clusters: for
$0\leq i\leq \log p$, a set formed by all the $p/2^i$ processors whose
indices share the most significant $i$ bits is called an
\emph{$i$-cluster}.  Hence, during an $i$-superstep each processor
communicates only with processors of its $i$-cluster. For the
communication within an $i$-cluster, parameter $\ell_i$ represents the
latency plus synchronization cost (in time units), while $g_i$
represents an inverse measure of bandwidth (in units of time per
message). By importing the notation adopted in the evaluation
model, we define the \emph {communication time} of 
an algorithm $\A$ on \execd\/ as,
\begin{equation}\label{eq:commtimeNO}
D_{\A}(n,p,\bfm{g},\bfm{\ell}) = \max_{I:|I|=n} \left\{\sum_{i=0}^{\log p-1}
\left(F^i_{\A}(I,p)g_i + S^i_{\A}(I)\ell_i\right)\right\}.
\end{equation}

Through the folding mechanism discussed above, any network-oblivious
algorithm $\A$ specified on \spec{v(n)} can be transformed into an
algorithm for \spec{p} with $p < v(n)$, hence into an algorithm for
\eval{p,\of} or \execd.  In this case, the quantities
$H_{\A}(n,p,\of)$ and $D_{\A}(n,p,\bfm{g},\bfm{\ell})$ denote,
respectively, the communication complexity and communication time
of the folded algorithm.  Moreover, since algorithms designed on the
evaluation model \eval{p,\of} or on the execution machine model
\execd\ can be regarded as algorithms for \spec{p}, once the
parameters $\of$ or $\bfm{g}$ and $\bfm{\ell}$ are fixed, we can also
analyze the communication complexities/times of their foldings on
smaller machines.  These relations among the models are crucial for
the effective exploitation of our framework.

The following definitions establish useful notions of optimality for
the two complexity measures introduced above relatively to the evaluation
and execution machine models. For each measure, optimality is defined
with respect to a class of algorithms, whose actual role
will be made clear later in the paper. Let $\mathscr{C}$ denote
a class of algorithms, solving a given problem  $\Pi$.

\begin{definition}\label{def:bopt}
Let $0 < \beta \leq 1$. 
An \evald-algorithm $\B \in \mathscr{C}$ is
$\beta$-\emph{optimal on \evald\ with respect to $\mathscr{C}$} 
if for each \evald-algorithm $\B' \in \mathscr{C}$ 
and for each $n$,
\[
H_{\B}(n,p,\sigma) \leq \dfrac{1}{\beta} H_{\B'}(n,p,\sigma)\,.
\]
\end{definition}
\begin{definition}\label{def:bopt2}
Let $0 < \beta \leq 1$. 
A \execd-algorithm $\B \in \mathscr{C}$ is $\beta$-\emph{optimal on
\execd\ with respect to $\mathscr{C}$} if for each \execd-algorithm $\B' \in
\mathscr{C}$ and for each $n$,
\[
D_{\B}(n,p,\bfm{g},\bfm{\ell}) \leq
\dfrac{1}{\beta} D_{\B'}(n,p,\bfm{g},\bfm{\ell})\,.
\]
\end{definition}

Note that the above definitions do not require $\beta$ to be a constant:
intuitively, larger values of $\beta$ correspond to higher degrees of
optimality.

\section{Optimality Theorem for Static Algorithms}\label{sec:optimality}
In this section, we show that, for a certain class of
network-oblivious algorithms, $\beta$-optimality in the evaluation
model, for suitable ranges of parameters $p$ and $\sigma$, translates
into $\beta'$-optimality in the execution machine model, for some
$\beta'=\BT{\beta}$ and suitable ranges of parameters $p$, $\bfm{g}$,
and $\bfm{\ell}$.  This result, which we refer to as optimality
theorem, holds under a number of restrictive assumptions;
nevertheless, it is applicable in a number of interesting case
studies, as illustrated in the subsequent sections. The optimality
theorem shows the usefulness of the intermediate evaluation model
since it provides a form of ``bootstrap,'' whereby from a given degree
of optimality on a family of machines we infer a related degree of
optimality on a much larger family.  It is important to remark that the
class of algorithms for which the optimality theorem holds includes
algorithms that are \emph{network-aware}, that is, whose code can make
explicit use of the architectural parameters of the model ($p$ and
$\of$, for the evaluation model, and $p$, $\bfm{g}$, and $\bfm{\ell}$,
for the execution machine model) for optimization purposes.

In a nutshell, the approach we follow hinges on the fact that both
communication complexity and communication time (Equations
\ref{eq:commcompNO} and \ref{eq:commtimeNO}) are expressed in terms of
quantities of the type $F^i_{\A}(I,p)$. If communication complexity is
low then these quantities must be low, whence communication time must
be low as well. Below, we discuss a number of obstacles to be faced
when attempting to refine the outlined approach into a rigorous
argument and how they can be handled.

A first obstacle arises whenever the performance functions are linear
combinations of other auxiliary metrics.  Unfortunately, worst-case
optimality of these metrics does not imply optimality of their linear
combinations (nor viceversa), since the worst case of different
metrics could be realized by different input instances. In the cases
of our interest, the ``losses'' incurred cannot be generally bounded
by constant factors. To circumvent this obstacle, we restrict our
attention to \emph{static algorithms}, defined by the property that
the following quantities are equal for all input instances of the same
size $n$: (i) the number of supersteps; (ii) the sequence of labels of
the various supersteps; and (iii) the set of source-destination pairs
of the messages exchanged in any individual superstep.  This
restriction allows us to overload the notation, writing $n$ instead of
$I$ in the argument of functions that become invariant for instances of
the same size, namely $L^i_{\A}(n)$, $S^i_{\A}(n)$, $h^s_{\A}(n,p)$,
and $F^i_{\A}(n,p)$. Likewise, the $\max$ operation becomes
superfluous and can be omitted in Equation~\ref{eq:commcompNO} and
Equation~\ref{eq:commtimeNO}. Static algorithms naturally arise in DAG
(Directed Acyclic Graph) computations. In a \emph{DAG algorithm}, for
every instance size $n$ there exists (at most) one DAG where each node
with indegree 0 represents an input value, while each node with
indegree greater than 0 represents a value produced by a unit-time
operation whose operands are the values of the node's predecessors
(nodes with outdegree 0 are viewed as outputs). The computation
requires the execution of all operations specified by the nodes,
complying with the data dependencies imposed by the arcs.\footnote{In
  the literature, DAG problems have also been referred to as pebble
  games~\cite[Section~10.1]{Savage98}.}

In order to prove the optimality theorem, we need a number of technical
results and definitions. Recall that folding can be employed to transform an
\eval{p,\of}-algorithm into an \eval{q,\of}-algorithm, for any $q < p$.  The
following lemma establishes a useful relation between the
communication metrics when folding is applied. 

\begin{lemma}\label{lem:general_up}
Let $\B$ be a static \eval{p,\of}-algorithm. For every $1 \leq j \leq
\log p$ and for every input size $n$, considering the folding of $\B$
on \eval{2^{j},0} we have
\[
\sum_{i=0}^{j-1} F^i_{\B}(n,2^j) \leq 
\dfrac{p}{2^{j}}\sum_{i=0}^{j-1} F^i_{\B}(n,p)\,.
\]
\end{lemma}

\begin{proof}
The lemma follows by observing that in every $i$-superstep,
with $i < j$, messages sent/destined by/to processor $P_k$ of \eval{2^{j},0},
with $0 \leq k < 2^{j}$, are a subset of those sent/destined
by/to the $p/2^{j}$ \eval{p,\of}-processors whose computations
are carried out by $P_k$.
\end{proof}

It is easy to come up with algorithms where the bound stated in the
above lemma is not tight. In fact, while in an $i$-superstep each
message must be exchanged between processors whose indices share \emph{at
  least} $i$ most significant bits, some messages which contribute to
$F^i_{\B}(n,p)$ may be exchanged between processors whose indices
share $j > i$ most significant bits, thus not contributing to
$F^i_{\B}(n,2^j)$. Motivated by this observation, we define below a
class of network-oblivious algorithms where a parameter $\alpha$
quantifies how tight the upper bound of Lemma~\ref{lem:general_up} is,
when considering their foldings on smaller machines. This parameter
will be employed to control the extent to which an optimality
guarantee in the evaluation model translates into an optimality
guarantee in the execution model.

\begin{definition}\label{def:wiseness}
A static network-oblivious algorithm $\A$ specified on
\spec{v(n)} is said to be $(\alpha,p)$-\emph{wise}, for some
$0 < \alpha \leq 1$ and $1 < p \leq v(n)$, if
considering the folding of 
$\A$ on \eval{2^{j},0} we have 
\[
\sum_{i=0}^{j-1} F^i_{\A}(n,2^j)
\geq \alpha\dfrac{p}{2^{j}} \sum_{i=0}^{j-1}F^i_{\A}(n,p)\, ,
\]
for every $1\leq j \leq \log p$ and every input size $n$. 
\end{definition}

\noindent
(We remark that in the above definition parameter $\alpha$ is not
necessarily a constant and can be made, for example, a function of
$p$.) As an example, a network-oblivious algorithm for \spec{v(n)}
where, for each $i$-superstep there is always at least one segment of
$v(n)/2^{i+1}$ virtual processors consecutively numbered starting from
$k \cdot (v(n)/2^{i+1})$, for some $k \geq 0$, each sending a number of
messages equal to the superstep degree to processors outside the
segment, is an $(\alpha, p)$-wise algorithm for each $1 < p \leq v(n)$
and $\alpha =1$.  However, $(\alpha, p)$-wiseness holds even
if the aforementioned communication scenario is realized only in an
average sense. Furthermore, consider a pair of values $\alpha'$ and
$p'$ such that $1 < p' \leq p$, and $1 < \alpha' \leq \alpha$. It is
easy to see that $(p/p')F^i_{\A}(n,p) \geq F^i_{\A}(n,p')$, for every
$0 \leq i < \log p'$, and this implies that a network-oblivious
algorithm which is $(\alpha,p)$-wise is also
$(\alpha',p')$-wise.

A final issue to consider is that the degrees of supersteps with
different labels contribute with the same weight to the communication
complexity while they contribute with different weights to the
communication time. The following lemma will help in bridging
this difference.
\begin{lemma}\label{lem:dominance}
For $m \geq 1$, let $\langle X_0, \ldots, X_{m-1} \rangle $ and $\langle Y_0,
\ldots, Y_{m-1} \rangle $ be two arbitrary sequences of real values, and
let $\langle f_0, \ldots, f_{m-1} \rangle$ be a non-increasing sequence of
non-negative values. If $\sum_{i=0}^{k-1} X_i \leq
\sum_{i=0}^{k-1} Y_i$, for every $1 \leq k \leq m$, then
\[
\sum_{i=0}^{m-1} X_i f_i \leq \sum_{i=0}^{m-1} Y_i f_i\,.
\]
\end{lemma}

\begin{proof}
By defining $S_{0}=0$ and $S_k = \sum_{j=0}^{k-1} (Y_j-X_j) \geq 0$,
for $1 \leq k \leq m$, we have:
\begin{gather*}
\sum_{i=0}^{m-1} f_i(Y_i-X_i) = \sum_{i=0}^{m-1} f_i(S_{i+1}-S_{i})
= \sum_{i=0}^{m-1} f_i S_{i+1} - \sum_{i=1}^{m-1}f_i S_{i} \geq\\
\geq \sum_{i=0}^{m-1} f_i S_{i+1} - \sum_{i=1}^{m-1}f_{i-1} S_{i}
= f_{m-1}S_{m} \geq 0\,.
\end{gather*}
\end{proof}

We are now ready to state and prove the optimality theorem.  Let
$\mathscr{C}$ denote a class of static algorithms solving a
problem $\Pi$, with the property that for any algorithm $\A \in
\mathscr{C}$ for $q$ processing elements, all of its foldings
on $q'$ processing elements, $2 \leq q' < q$, also belong to
$\mathscr{C}$.
\begin{theorem}[Optimality theorem]\label{thm:optimality} \sloppy
Let $\A \in \mathscr{C}$ be network-oblivious and
$(\alpha,p^\star)$-wise, for some $\alpha \in (0,1]$ and a power of
two $p^\star$. Let also $(\of^m_0, \ldots, \of^m_{\log
    p^{\star}-1})$ and $(\of^M_0, \ldots, \of^M_{\log p^{\star}-1})$
  be two vectors of non-negative values, with 
  $\of^m_j \leq \of^M_j$, for every $0 \leq j < \log p^{\star}$. 
  If $\A$ is $\beta$-optimal on
  \eval{2^{j},\of} w.r.t.\ $\mathscr{C}$, for 
  $\of_{j-1}^m \leq \of \leq \of^M_{j-1}$ and $1 \leq j \leq \log
  p^\star$, then, for
  every power of two $p \leq p^\star$, $\A$ is
  $\alpha\beta/(1+\alpha)$-optimal on \execd\ w.r.t. $\mathscr{C}$ 
  as long as:
\begin{itemize}
 \item $g_{i} \geq g_{i+1}$ and $\ell_{i}/g_{i} \geq \ell_{i+1}/g_{i+1}$,
  for $0 \leq i < \log p-1$; 
 \item $\max_{1 \leq k \leq \log p}\{\of^m_{k-1} 2^{k} /p\}
  \leq \ell_{i}/g_{i} \leq \min_{1 \leq k \leq \log p}\{\of^M_{k-1} 2^{k} /p\}$, 
  for $0 \leq i < \log p$.\footnote{%
Note that in order to allow for a nonempty range of values for the ratio
$\ell_{i}/g_{i}$, the $\of^m$ and $\of^M$ vectors  must be such
that $\max_{1 \leq k \leq \log p}\{\of^m_{k-1} 2^{k} /p\}
 \leq \min_{1 \leq k \leq \log p}\{\of^M_{k-1} 2^{k} /p\}$. This will
always be the case for the applications discussed in the next section.}
\end{itemize}
\end{theorem}

\begin{proof}
Fix the value $p$ and the vectors $\bfm{g}$ and $\bfm{\ell}$ so to
satisfy the hypotheses of the theorem, and consider a
\execd-algorithm $\C \in \mathscr{C}$. By the $\beta$-optimality of
$\A$ on the evaluation model $M(2^{j},\psi p/2^{j})$, for each $1 \leq
j \leq \log p$ and $\psi$ such that $\sigma^m_{j-1} \leq \psi p/ 2^{j}
\leq \sigma^M_{j-1}$, we have
\[
H_{\A}\left(n,2^j, \dfrac{\psi p}{2^{j}}\right)\leq \dfrac{1}{\beta}
H_{\C}\left(n,2^j, \dfrac{\psi p}{2^{j}}\right)
\]
since $\C$ can be folded into an algorithm for $M(2^{j},\psi p/2^{j})$,
still belonging to $\mathscr{C}$.
By the definition of communication complexity it follows that
\[
\sum_{i=0}^{j-1} \left(F_{\A}^i(n,2^{j}) + S_{\A}^i(n) \dfrac{\psi
p}{2^{j}} \right) \leq
\dfrac{1}{\beta} \sum_{i=0}^{j-1} \left(F_{\C}^i(n,2^{j}) +
S_{\C}^i(n)\dfrac{\psi p}{2^{j}}\right),
\]
and then, by applying Lemma~\ref{lem:general_up} to the right side of the above
inequality, we obtain
\begin{equation}\label{eq:partial_sums}
\sum_{i=0}^{j-1} \left(F_{\A}^i(n,2^{j}) + S_{\A}^i(n) \dfrac{\psi
p}{2^{j}} \right) \leq
\dfrac{1}{\beta} \sum_{i=0}^{j-1} \left(\dfrac{p}{2^{j}}F_{\C}^i(n,p) +
S_{\C}^i(n) \dfrac{\psi p}{2^{j}}\right).
\end{equation}

Define $\psi^m_p = \max_{1 \leq k \leq \log p}\{\of^m_{k-1} 2^{k} /p\}$
and $\psi^M_p = \min_{1 \leq k \leq \log p}\{\of^M_{k-1} 2^{k} /p\}$.
The condition imposed by the theorem on the ratio $\ell_{i}/g_{i}$
implies that $\psi^m_p\leq \psi^M_p$, hence, by definition of these
two quantities, we have that 
$\sigma^m_{j-1} 2^{j} /p\leq \psi^m_p, \psi_p^M
\leq \sigma^M_{j-1} 2^{j} /p$. 

Let us first set $\psi=\psi_p^M$ in Inequality~\ref{eq:partial_sums},
and note that, by the above observation, $\sigma^m_{j-1} \leq \psi_p^M
p/ 2^{j} \leq \sigma^M_{j-1}$.  By multiplying both terms of the
inequality by $2^{j} / (\psi_p^M p)$, and by exploiting the
non-negativeness of the $F_{\A}^i(n,2^{j})$ terms, we obtain
\[
\sum_{i=0}^{j-1} S_{\A}^i(n)  \leq
\dfrac{1}{\beta} \sum_{i=0}^{j-1} \left(\dfrac{F_{\C}^i(n,p)}{\psi_p^M} +
S_{\C}^i(n) \right).
\]
Next, we make $\log p$ applications of Lemma~\ref{lem:dominance}, one for
each $j = 1,2,\dots,\log p$, by setting $m = j$, $X_i = S_{\A}^i(n)$,
$Y_i = (1 / \beta) \left(F_{\C}^i(n,p)/\psi_p^M + S_{\C}^i(n)\right)$,
and $f_i = \ell_i/g_i$. This gives
\[
\sum_{i=0}^{j-1} S_{\A}^i(n)  \dfrac{\ell_i}{g_i} \leq
\dfrac{1}{\beta} \sum_{i=0}^{j-1} \left(F_{\C}^i(n, p)\dfrac{\ell_i}{\psi_p^M
g_i} +
S_{\C}^i(n) \dfrac{\ell_i}{g_i}  \right ),
\]
for $1 \leq j \leq \log p$. Since, by hypothesis, $\ell_i / g_i \leq
\psi_p^M$, for each $0 \leq i < \log p$, we have $\ell_i / \psi_p^M
g_i \leq 1$, hence we can write
\begin{equation}\label{eq:latencies}
\sum_{i=0}^{j-1} S_{\A}^i(n)  \dfrac{\ell_i}{g_i} \leq
\dfrac{1}{\beta} \sum_{i=0}^{j-1} \left(F_{\C}^i(n,p) +
S_{\C}^i(n) \dfrac{\ell_i}{g_i} \right ),
\end{equation}
for $1 \leq j \leq \log p$.

Now, let us set $\psi = \psi_p^m $ in
Inequality~\ref{eq:partial_sums}, which, again, guarantees
$\sigma^m_{j-1} \leq \psi_p^m p/ 2^{j} \leq \sigma^M_{j-1}$. By exploiting
the wiseness of $\A$ in the left side and the non-negativeness of
$S^i_{\A}(n)$, we obtain
\[
\sum_{i=0}^{j-1} \alpha\dfrac{p}{2^j} F_{\A}^i(n,p) \leq
\dfrac{1}{\beta} \sum_{i=0}^{j-1} \left(\dfrac{p}{2^{j}}F_{\C}^i(n,p) +
S_{\C}^i(n) \dfrac{\psi_p^m  p}{2^{j}}\right).
\]
By multiplying both terms by $2^j / (p\alpha)$ and observing that,
by hypothesis, $\psi_p^m \leq \ell_i/g_i$,
for each $0 \leq i < \log p$, we get
\begin{equation}\label{eq:traffics}
\sum_{i=0}^{j-1} F_{\A}^i(n, p)  \leq
\dfrac{1}{\alpha\beta} \sum_{i=0}^{j-1} \left(F_{\C}^i(n,p) +
S_{\C}^i(n) \dfrac{\ell_i}{g_i}\right).
\end{equation}
Summing Inequality~\ref{eq:latencies} with Inequality~\ref{eq:traffics} yields
\[
\sum_{i=0}^{j-1} \left(F_{\A}^i(n,p) + S_{\A}^i(n)
\dfrac{\ell_i}{g_i}\right) \leq
\dfrac{1 + \alpha}{\alpha\beta}  \sum_{i=0}^{j-1} \left( F_{\C}^i(n,p) +
S_{\C}^i(n) \dfrac{\ell_i}{g_i}\right),
\]
for $1 \leq j \leq \log p$. Applying Lemma~\ref{lem:dominance}
with $m = \log p$, $X_i = F_{\A}^i(n,p) + S_{\A}^i(n) \ell_i /g_i$,
$Y_i = (1 + \alpha)/(\alpha \beta)\left(F_{\C}^i(n,p) + S_{\C}^i(n) \ell_i/g_i\right)$,
and $f_i = g_i$ yields
\[
\sum_{i=0}^{\log p -1} \left(F_{\A}^i(n,p) g_i + S_{\A}^i(n)
\ell_i \right) \leq
\dfrac{1 + \alpha}{\alpha\beta} \sum_{i=0}^{\log p -1} \left(F_{\C}^i(n,p)
g_i +
S_{\C}^i(n) \ell_i \right).
\]
Then, by definition of communication time we have
\[
D_{\A}\left(n,p, \bfm{g}, \bfm{\ell}\right)\leq 
\dfrac{1 + \alpha}{\alpha\beta} D_{\C}\left(n,p, \bfm{g}, \bfm{\ell}\right),
\]
and the theorem follows.
\end{proof}

A few remarks are in order regarding the above optimality theorem. 
The complexity metrics adopted in this paper target exclusively
interprocessor communication, thus a (sequential) network-oblivious
algorithm specified on \spec{v} but using only one of the virtual
processors would clearly be optimal with respect to these metrics. For
meaningful applications of the theorem, the class $\mathscr{C}$ must
be suitably defined to exclude such degenerate cases and to contain
algorithms where the work is sufficiently well balanced among the
processing elements.

Note that the theorem requires that both the $g_i$'s and $\ell_i /g_i$'s
form non-increasing sequences. The assumption is rather natural
since it reflects the fact that larger submachines exhibit 
more expensive communication (hence, a larger $g$ parameter) and
larger network capacity (hence, a larger $\ell/g$ ratio).

Some of the issues encountered in establishing the optimality theorem
have an analog in the context of memory hierarchies. For example, time
in the Hierarchical Memory Model (HMM) can be linked to I/O complexity
as discussed in~\cite{AggarwalACS87}, so that optimality of the latter
for different cache sizes implies the optimality of the former for
wide classes of functions describing the access time to different
memory locations. Although, to the best of our knowledge, the question
has not been explicitly addressed in the literature, a careful
inspection of the arguments of~\cite{AggarwalACS87} shows that some
restriction to the class of algorithms is required, to guarantee that
the maximum value of the I/O complexity for different cache sizes is
simultaneously reached for the same input instance. (For example, the
optimality of HMM time does not follow for the class of arbitrary
comparison-based sorting algorithms, since the known I/O complexity
lower bound for this problem~\cite{AggarwalV88} may not be
simultaneously reachable for all relevant cache sizes.) Moreover, the
monotonicity we have assumed for the $g_i$ and the $\ell_i/g_i$
sequences has an analog in the assumption that the function used in
\cite{AggarwalACS87} to model the memory access time is polynomially
bounded.

In the cache-oblivious framework the equivalent of our optimality
theorem requires algorithms to satisfy the \emph{regularity
condition}~\cite[Lemma~6.4]{FrigoLPR12}, which requires that the
number of cache misses decreases by a constant factor when the cache
size is doubled. On the other hand, our optimality theorem gives the
best bound when the network-oblivious algorithm is $(\BT{1},p)$-wise,
that is, when the communication complexity decreases by a constant
factor when the number of processors is doubled.  Although the
regularity condition and  wiseness cannot be formalized in a similar
fashion due to the significant differences between the cache- and
network-oblivious frameworks, we observe that both assumptions require
the oblivious algorithms to react seamlessly and smoothly  to small
changes of the machine parameters.

\section{Algorithms for Key Problems}\label{sec:algorithms}

In this section, we illustrate the use of the proposed framework by
developing efficient network-oblivious algorithms for a number of
fundamental computational problems: matrix multiplication
(Subsection~\ref{sec:alg_mult}), Fast Fourier Transform
(Subsection~\ref{sec:alg_fft}), and sorting
(Subsection~\ref{sec:alg_sort}).  All of our algorithms exhibit
$\BT{1}$-optimality on the D-BSP for wide ranges of the machine
parameters. In Subsection~\ref{sec:stencil}, we also present
network-oblivious algorithms for stencil computations. 
These latter algorithms run efficiently on the D-BSP although they do
not achieve $\BT{1}$-optimality, which appears to be a hard challenge in this case.
In Subsection~\ref{sec:alg_broadcast}, we also establish a negative
result by proving that there cannot exist a network-oblivious algorithm
for broadcasting which is simultaneously $\BT{1}$-optimal on two
sufficiently different \evald\ machines.

As prescribed by our framework, the performance of the
network-oblivious algorithms on the D-BSP is derived by analyzing
their performance on the evaluation model. Optimality is assessed with
respect to classes of algorithms where the computation is not
excessively unbalanced among the processors, namely, algorithms where
an individual processor cannot perform more than a constant fraction
of the total minimum work for the problem.  For this purpose, we
exploit some recent lower bounds which rely on mild assumptions on
work distributions and strengthen previous bounds based on stronger
assumptions~\cite{ScquizzatoS14}. Finally, we want to stress that all
of our algorithms are also work optimal.

\subsection{Matrix Multiplication}\label{sec:alg_mult}
The \emph{$n$-MM problem} consists of multiplying two $\sqrt{n} \times
\sqrt{n}$ matrices, $A$ and $B$, using only semiring operations. For
convenience, we assume that $n$ is an even power of 2 (the general case
requires minor yet tedious modifications).  A result in~\cite{Kerr70}
shows that any static algorithm for the $n$-MM problem
which uses only semiring operations must compute all $n^{3/2}$
\emph{multiplicative terms}, that is the products $A[i,k]\cdot
B[k,j]$, with $0 \leq i,j,k \leq \sqrt{n}$.

Let $\mathscr{C}$ denote the class of static algorithms for the $n$-MM
problem such that any $\A \in \mathscr{C}$ for $q$ processing
elements satisfies the following properties: (i) no entry of $A$ or $B$
is initially replicated (however, the entries of $A$ and $B$ are allowed to be initially
distributed among the processing elements in an arbitrary fashion);
(ii) no processing element computes more than $n^{3/2}/\min\{q,11^3\}$
multiplicative terms; (iii) all of the foldings of $\A$ on $q'$
processing elements, $2 \leq q' < q$, also belong to $\mathscr{C}$.
The following lemma establishes a lower bound on the communication
complexity of the algorithms in $\mathscr{C}$.

\begin{lemma}\label{MMlowerbound}
The communication complexity of any $n$-MM algorithm in $\mathscr{C}$
when executed on \evald\ is $\BOM{n/p^{2/3}+\of}$.
\end{lemma}

\begin{proof}
The bound for $\of=0$ is proved in~\cite[Theorem~2]{ScquizzatoS14}, and it clearly
extends to the case $\of>0$. The additive $\of$ term follows since at least
one message is sent by some processing element.
\end{proof}

We now describe a static network-oblivious algorithm for the $n$-MM problem, 
which follows from the parallelization of the respective cache-oblivious
algorithm~\cite{FrigoLPR12}. Then, we prove its optimality in the evaluation
model, for wide ranges of the parameters, and in the execution model through
the optimality theorem. The algorithm is specified on \spec{n}, and requires
that the input and output matrices be evenly distributed among the $n$ VPs.
We denote with $A$, $B$ and $C$ the two input matrices and the output
matrix respectively, and with $A_{hk}$, $B_{hk}$ and $C_{hk}$, with
$0 \leq h,k \leq 1$, their four quadrants. The network-oblivious
algorithm adopts the following recursive strategy:
\begin{enumerate}
\item
Partition the VPs into eight segments $S_{hk\ell}$, with $0 \leq h, k, \ell \leq
1$, containing the same number of consecutively numbered VPs. Replicate and
distribute the inputs so that the entries of $A_{h\ell}$ and $B_{\ell k}$ are
evenly spread among the VPs in $S_{hk\ell}$.
\item
In parallel, for each $0 \leq h, k, \ell \leq 1$, compute recursively the
product $M_{hk\ell} = A_{h\ell} \cdot B_{\ell k}$ within $S_{hk\ell}$.
\item
In parallel, for each $0 \leq i,j < \sqrt{n}$, the VP responsible for
$C[i,j]$ collects $M_{hk0}[i',j']$ and $M_{hk1}[i',j']$, with
$h = \lfloor 2i/\sqrt{n}\rfloor$, $k = \lfloor 2j/\sqrt{n}\rfloor$, 
$i' = i \mod (\sqrt{n}/2)$ and $j' = j \mod (\sqrt{n}/2)$, and
computes $C[i,j] = M_{hk0}[i',j']+M_{hk1}[i',j']$.
\end{enumerate}
At the $i$-th recursion level, with $0 \leq i \leq (\log n)/3$, $8^i$
$(n/4^i)$-MM subproblems are solved by distinct \spec{n/8^i}'s formed
by distinct segments of VPs.  The recursion stops at $i = (\log n)/3$
when each VP solves sequentially an $n^{1/3}$-MM subproblem.  It is
easy to see, by unfolding the recursion, that the algorithm comprises
a constant number of $3i$-supersteps for every $0 \leq i < (\log n)/3$,
where each VP sends/receives $\BO{2^i}$ messages. In order to enforce
wiseness, without worsening the asymptotic communication complexity
and communication time exhibited by the algorithm in the evaluation
and execution machine models, we may assume that, in each $3i$-superstep,
$\text{VP}_j$ sends $2^i$ dummy messages to $\text{VP}_{j+n/2^{3i+1}}$,
for $0 \leq j < n/2^{3i+1}$.

\begin{theorem}\label{mm-omm}
The communication complexity of the above $n$-MM network-oblivious algorithm 
when executed on \evald\ is
\[
H_{\textnormal{MM}}(n,p,\of) = \BO{\dfrac{n}{p^{2/3}} + \of \log p},
\]
for every $1 < p \leq n$ and $\of \geq 0$. The algorithm is $(\BT{1},n)$-wise
and $\BT{1}$-optimal with respect to $\mathscr{C}$ on any \eval{p,\of} with
$1 < p \leq n$ and $\of = \BO{n/(p^{2/3}\log p)}$.
\end{theorem}

\begin{proof}
From the previous discussion, we conclude that
\[
H_{\textnormal{MM}}(n,p,\of) = \BO{\sum_{i=0}^{(\log p)/3} \left(\frac{n2^i}{p} + \of\right)}
=
\BO{\dfrac{n}{p^{2/3}} + \of \log p}.
\]
As anticipated, the wiseness is guaranteed by the dummy messages
introduced in each superstep. Finally, it is easy to see that the
algorithm satisfies the three requirements for belonging to
$\mathscr{C}$, hence its optimality follows from Lemma~\ref{MMlowerbound}.
\end{proof}

\begin{corollary}\label{lem:opt_nmm}
The above $n$-MM network-oblivious algorithm is $\BT{1}$-optimal with
respect to $\mathscr{C}$ on any \execd\ machine with $1 < p \leq n$,
non-increasing $g_i$'s and $\ell_i/g_i$'s, and $\ell_0/g_0 = \BO{n/p}$.
\end{corollary}

\begin{proof}
Since the network-oblivious algorithm is $(\BT{1},n)$-wise and belongs
to $\mathscr{C}$, the corollary follows by plugging $p^\star = n$,
$\of_i^m = 0$, and $\of_i^M = \BT{n/((i+1)2^{2i/3})}$ into
Theorem~\ref{thm:optimality}.
\end{proof}

\subsubsection{Space-Efficient Matrix Multiplication}
Observe that the network-oblivious algorithm described above incurs an
$\BO{n^{1/3}}$ memory blow-up per VP. As described below, the
recursive strategy can be modified so to incur only a constant memory
blow-up, at the expense of an increased communication complexity. The
resulting network-oblivious algorithm turns out to be $\BT{1}$-optimal
with respect to the class of algorithms featuring constant memory
blow-up.

We assume, as before, that the entries of $A,B$ and $C$ be evenly
distributed among the VPs. The VPs are (recursively) divided into four
segments which solve the eight $(n/4)$-MM subproblems in two rounds:
in the first round, the segments compute $A_{00} \cdot B_{00}$,
$A_{01} \cdot B_{11}$, $A_{11} \cdot B_{10}$ and $A_{10} \cdot B_{01}$
(one product per segment), while in the second round they compute
$A_{01} \cdot B_{10}$, $A_{00} \cdot B_{01}$, $A_{10} \cdot B_{00}$ and
$A_{11} \cdot B_{11}$ (again, one product per segment).  The recursion
ends when each VP solves sequentially an $1$-MM subproblem. By
unfolding the recursion, we get that for every $0 \leq i < \log n /2$,
the algorithm executes $\BT{2^i}$ $2i$-supersteps where each VP
sends/receives $\BT{1}$ messages. At any time each VP contains only
$\BO{1}$ matrix entries, but the recursion requires it to handle a
stack of $\BO{\log n}$ entries. However, it is easy to see that only a
constant number of bits are needed for each stack entry, hence, under
the natural assumption that each matrix entry occupies a constant
number of $\BT{\log n}$-bit words, the entire stack at each VP
requires storage proportional to $\BO{1}$ matrix entries. Therefore,
the algorithm incurs only a constant memory blow-up. As before, the
algorithm can be easily made $(\BT{1},n)$-wise by adding suitable
dummy messages and, when executed on \evald, its communication
complexity becomes $\BO{n/\sqrt{p} + \of \sqrt{p}}$.

Let $\mathscr{C}'$ denote the class of static algorithms for the
$n$-MM problem such that any $\A \in \mathscr{C}'$ for $q$ processing
elements satisfies the following properties: (i) the local storage
required at each processing element is $\BO{n/q}$; (ii) all of the
foldings of $\A$ on $q'$ processing elements, $2 \leq q' < q$, also
belong to $\mathscr{C}'$. Since it is proved in~\cite{IronyTT04} that
any $n$-MM algorithm in $\mathscr{C}'$ when running on \eval{p,0} must
exhibit an $\BOM{n/\sqrt{p}}$ communication complexity, the above
network-oblivious algorithm is $\BT{1}$-optimal with respect to
$\mathscr{C}'$ on any \eval{p,\of} with $1 < p \leq n$ and $\of =
\BO{n/p}$.  Consequently, Theorem~\ref{thm:optimality} yields
optimality of the algorithm on any \execd\ machine with $1 < p \leq n$,
non-increasing $g_i$'s and $\ell_i/g_i$'s, and $\ell_0/g_0 =
\BO{n/p}$.

\subsection{Fast Fourier Transform}\label{sec:alg_fft}
The \emph{$n$-FFT problem} consists of computing the Discrete Fourier
Transform of $n$ values using the $n$-input FFT DAG, where a vertex is
a pair $\langle w,l \rangle$, with $0 \leq w <n$ and $0 \leq l < \log n$,
and there exists an arc between two vertices $\langle w,l \rangle$ and
$\langle w',l' \rangle$ if $l'=l+1$, and either $w$ and $w'$ are
identical or their binary representations differ exactly in the
$l$-th bit~\cite{Leighton92}. 

Let $\mathscr{C}$ denote the class of static algorithms for the
$n$-FFT problem such that any $\A \in \mathscr{C}$ for $q$ processing
elements satisfies the following properties: (i) each DAG node is
evaluated exactly once (i.e., recomputation is not allowed); (ii)
no input value is initially replicated; (iii) no processing element
computes more than $\epsilon n \log n$ DAG nodes, for some constant
$0 < \epsilon < 1$; (iv) all of the foldings of $\A$ on $q'$ processing
elements, $2 \leq q' < q$, also belong to $\mathscr{C}$. Note that,
as in the preceding subsection, the class of algorithms we are considering
makes no assumptions on the input and output distributions.
we make no assumptions on the input and output distributions. The
following lemma establishes a lower bound on the communication complexity
of the algorithms in $\mathscr{C}$.

\begin{lemma}\label{FFTlowerbound}
The communication complexity of any $n$-FFT algorithm in $\mathscr{C}$
when executed on \evald\ is $\BOM{(n \log n)/(p \log(n/p)) + \of}$.
\end{lemma}

\begin{proof}
The bound for $\of=0$ is proved in~\cite[Theorem~11]{ScquizzatoS14}, and it clearly
extends to the case $\of>0$. The additive $\of$ term follows since at least
one message is sent by some processing element.
\end{proof}

We now describe a static network-oblivious algorithm for the $n$-FFT
problem, and then prove its optimality in the evaluation and execution
models. The algorithm is specified on \spec{n}, and exploits the well-known
decomposition of the FFT DAG into two sets of $\sqrt{n}$-input FFT subDAGs,
with each set containing $\sqrt{n}$ such subDAGs~\cite{AggarwalCS87}.
For simplicity, in order to ensure integrality of the quantities involved,
we assume $n = 2^{2^k}$ for some integer $k \geq 0$. We assume that at
the beginning the $n$ inputs are evenly distributed among the $n$ VPs.
In parallel, each of the $\sqrt{n}$ segments of $\sqrt{n}$ consecutively
numbered VPs computes recursively the assigned subDAG. Then, the outputs
of the first set of subDAGs are permuted in a 0-superstep so to distribute
the inputs of each subDAGs of the second set among the VPs of a distinct
segment. The permutation pattern is equivalent to the transposition of a
$\sqrt{n}\times \sqrt{n}$ matrix. Finally, each segment computes 
recursively the assigned subDAG.

At the $i$-th recursion level, with $i \geq 0$, $n^{1-1/2^i}$
$n^{1/2^i}$-FFT subproblems are solved by $n^{1-1/2^i}$ \spec{n^{1/2^i}}
models formed by distinct segments of VPs.  The recurrence stops at
$i = \log \log n$ when each segment of two VPs computes a 2-input subDAG.
It is easy to see, by unfolding the recursion, that the algorithm comprises
$\BO{2^i}$ supersteps with label $(1-1/2^i)\log n$, for each $0 \leq i <
\log \log n$, where each VP sends/receives $\BO{1}$ messages.  As before,
in order to enforce wiseness without affecting the algorithm's asymptotic
performance, we assume that in each $(1-1/2^i)\log n$-superstep,
$\text{VP}_j$ sends a dummy message to $\text{VP}_{j+n^{1/2^i}/2}$,
for each $0 \leq j < n^{1/2^i}/2$.

\begin{theorem}\label{fft-ub}
The communication complexity of the above $n$-FFT network-oblivious
algorithm when executed on \eval{p,\of} is
\[
H_{\textnormal{FFT}}(n,p,\of) =
\BO{\left(\dfrac{n}{p}+\of\right)\dfrac{\log n}{\log(n/p)}},
\]
for every $1 < p \leq n$ and $\of \geq 0$. The algorithm is $(\BT{1},n)$-wise
and $\BT{1}$-optimal with respect to $\mathscr{C}$ on any \eval{p,\of} with
$1 < p \leq n$ and $\of = \BO{n/p}$.
\end{theorem}

\begin{proof}
From the previous discussion, we conclude that
\[
H_{\textnormal{FFT}}(n,p,\of) 
= \BO{\sum_{i=0}^{\log (\log n/\log(n/p))} 2^i\left(\frac{n}{p} + \of\right)}
= \BO{\left(\dfrac{n}{p}+\of\right)\dfrac{\log n}{\log (n/p)}}.
\]
The wiseness is ensured by the dummy messages, and since the algorithm
satisfies the requirements for belonging to $\mathscr{C}$, its
optimality follows from Lemma~\ref{FFTlowerbound}.
\end{proof}

We now apply Theorem~\ref{thm:optimality} to show that the
network-oblivious algorithm is $\BT{1}$-optimal on the D-BSP for wide
ranges of the machine parameters.

\begin{corollary}
The above $n$-FFT network-oblivious algorithm is $\BT{1}$-optimal with
respect to $\mathscr{C}$ on any \execd\ machine with $1 < p \leq n$,
non-increasing $g_i$'s and $\ell_i/g_i$'s, and $\ell_0/g_0 = \BO{n/p}$.
\end{corollary}

\begin{proof}
Since the network-oblivious algorithm is $(\BT{1},n)$-wise and belongs
to $\mathscr{C}$, we get the claim by plugging $p^\star = n$, $\of_i^m = 0$,
and $\of_i^M = \BT{n/2^i}$ in Theorem~\ref{thm:optimality}.
\end{proof}

We observe that although we described the network-oblivious algorithm
assuming $n = 2^{2^k}$, in order to ensure integrality of the quantities
involved, the above results can be generalized to the case of $n$
arbitrary power of two. In this case, the FFT DAG is recursively
decomposed into a set of $2^{\lfloor \log \sqrt{n} \rfloor}$-input
FFT subDAGs and a set of $n/2^{\lfloor \log \sqrt{n} \rfloor}$-input
FFT subDAGs. The optimality of the resulting algorithm in both the
evaluation and execution machine models can be proved in a similar
fashion as before.

\subsection{Sorting}\label{sec:alg_sort}
The \emph{$n$-sort problem} requires to label $n$ (distinct) input
keys with their ranks, using only comparisons, where the \emph{rank}
of a key is the number of smaller keys in the input sequence.

Let $\mathscr{C}$ denote the class of static algorithms for the
$n$-sort problem such that any $\A \in \mathscr{C}$ for $q$
processing elements satisfies the following properties: (i) initially,
no input key is replicated and, during the course of the algorithm,
only a constant number of copies per key are allowed at any time; (ii) no
processing element performs more than $\epsilon n \log n$ comparisons,
for an arbitrary constant $0<\epsilon<1$; (iii) all of the foldings of
$\A$ on $q'$ processing elements, $2 \leq q' < q$, also belong to
$\mathscr{C}$.  We make no assumptions on how the keys are distributed 
among the processing elements at the beginning and at the end
of the algorithm. The following lemma establishes a lower bound
on the communication complexity of the algorithms in $\mathscr{C}$.

\begin{lemma}\label{Sortlowerbound}
The communication complexity of any $n$-sort algorithm in $\mathscr{C}$
when executed on \evald\ is $\BOM{(n \log n)/(p \log(n/p))+\of}$.
\end{lemma}

\begin{proof}
The bound for $\of=0$ is proved in~\cite[Theorem~8]{ScquizzatoS14}, and it clearly
extends to the case $\of>0$. The additive $\of$ term follows since at least
one message is sent by some processing element.
\end{proof}

We now present a static network-oblivious algorithm for the $n$-sort
problem, and then prove its optimality in the evaluation and execution
models. The algorithm implements a recursive version of the
\emph{Columnsort} strategy, as described in~\cite{Leighton85}.
Consider the $n$ input keys as an $r \times s$ matrix, with $r\cdot
s=n$ and $r\geq s^2$. Columnsort is
organized into seven \emph{phases} numbered from 1 to 7. During Phases
1, 3, 5 and 7 the keys in each column are sorted recursively (in Phase
5 adjacent columns are sorted in reverse order). During Phases 2, 4, 6 
and 8 the keys of the matrix are permuted: in Phase 2 (resp., 4) a transposition (resp.,
diagonalizing permutation~\cite{Leighton85}) of the $r\times s$ matrix is performed maintaining the
$r\times s$ shape; in Phase 6 (resp., 8) an $r/2$-cyclic shift (resp., the reverse of the
$r/2$-cyclic shift) is done.\footnote{In the original paper~\cite{Leighton85}, the shift in Phase 6
is not cyclic: a new column is added containing the $r/2$ overflowing keys and $r/2$ large dummy
keys, while the first column is filled with $r/2$ small dummy keys. However, it is easy to see
that a cyclic shift suffices if the first $r/2$ keys in the first column are considered
smaller than the last $r/2$ keys.}
Columnsort can be implemented on \spec{n} as follows. For convenience,
assume that $n=2^{(3/2)^d}$ for some integer $d \geq 0$, and set
$r=n^{2/3}$ and $s=n/r$ (the more general case is discussed later).
The algorithm starts with the input keys evenly distributed among the
$n$ VPs. In the odd phases the keys of each column are evenly
distributed among the VPs of a distinct segment of $r$ consecutively
numbered VPs, which form an independent \spec{r}. Then, each segment
solves recursively the subproblem corresponding to the column it
received. The even phases entail a constant number of $0$-supersteps
of constant degree.  At the $i$-th recursion level, with $0\leq i \leq
\log_{3/2} \log n$, each segment of $n^{(2/3)^i}$ consecutively
numbered VPs forming an independent \spec{n^{(2/3)^i}} solves $4^i$
subproblems of size $n^{(2/3)^i}$.  The recurrence stops at $i =
\log_{3/2} \log n$ when each VP solves, sequentially, a subproblem of
constant size.  It is easy to see, by unfolding the recursion, that
the algorithm consists of $\BT{4^i}$ supersteps with label
$(1-2/3^i)\log n$, for each $0\leq i < \log_{3/2} \log n$, where each
VP sends/receives $\BO{1}$ messages.  As before, in order to enforce
wiseness, without affecting the algorithm's asymptotic performance, we
assume that in each $(1-(2/3)^i)\log n$-superstep, $\text{VP}_j$ sends
a dummy message to $\text{VP}_{j+n^{(2/3)^i}/2}$, for each $0\leq j <
n^{(2/3)^i}/2$.

\begin{theorem}\label{thm-sort}
The communication complexity of the above network-oblivious algorithm
for $n$-sort when executed on \eval{p,\of} is
\[
H_{\textnormal{sort}}(n,p,\of)=
\BO{\left(\dfrac{n}{p}+\of\right) 
    \left(\dfrac{\log n}{\log(n/p)}\right)^{\log_{3/2}4}},
\]
for every $1 < p \leq n$ and $\of \geq 0$.
The algorithm is $(\BT{1},n)$-wise and 
it is $\BT{1}$-optimal with respect to $\mathscr{C}$ on any
\eval{p,\of} with  $p = \BO{n^{1-\delta}}$, for any 
arbitrary constant $\delta \in (0,1)$,
and $\of \geq 0$.
\end{theorem}

\begin{proof}
By the previous description of the algorithm, we conclude that
\[
H_{\textnormal{sort}}(n,p,\of) =
\BO{\sum_{i=0}^{\log_{3/2} (\log n/\log(n/p))} 4^i\left(\frac{n}{p} + \of\right)}
= \BO{\left(\dfrac{n}{p}+\of\right) \left(\dfrac{\log n}{\log(n/p)}\right)^{\log_{3/2}4}}.
\]
The wiseness is guaranteed by the dummy messages. Since the algorithm
satisfies the three requirements to be in $\mathscr{C}$, its
optimality follows from Lemma~\ref{Sortlowerbound}.
\end{proof}

\begin{corollary}\label{cor-sort}
The above $n$-sort network-oblivious algorithm is $\BT{1}$-optimal
with respect to $\mathscr{C}$ on any \execd\ machine with $p =
\BO{n^{1-\delta}}$, for some arbitrary constant $\delta \in
(0,1)$, and non-increasing $g_i$'s and $\ell_i/g_i$'s.
\end{corollary}

\begin{proof}
Since the network-oblivious algorithm is $(\BT{1},n)$-wise and belongs
to $\mathscr{C}$, we get the claim by plugging $p^\star = n$, $\of_i^m
= 0$, and $\of_i^M = +\infty$ in Theorem~\ref{thm:optimality}.
\end{proof}

Consider now the more general case when
$n$ is an arbitrary power of two.  Now, the input keys must be regarded
as the entries of an $r \times s$ matrix, where $r$ is the smallest
power of two greater than or equal to $n^{2/3}$. Simple, yet tedious,
calculations show that the results stated in Theorem~\ref{thm-sort}
and Corollary~\ref{cor-sort} continue to hold in this case.

Finally, we remark that the above network-oblivious sorting algorithm
turns out to be $\BT{1}$-optimal on any \execd, as long as $p =
\BO{n^{1-\delta}}$ for constant $\delta$, with respect to a wider
class of algorithms which satisfy the requirements (i), (ii), and
(iii), specified above for $\mathscr{C}$, but need not be static.
By applying the lower bound for sorting in~\cite{ScquizzatoS14} on
two processors, it is easy to show that $\BOM{n}$ messages must cross
the bisection for this class of algorithms. Therefore, we get an
$\BOM{g_0 n/p}$ lower bound on the communication time on \execd, which
is matched by our network-oblivious algorithm.

\subsection{Stencil Computations}\label{sec:stencil}
A \emph{stencil} defines the computation of any element in a $d$-dimensional
spatial grid at time $t$ as a function of neighboring grid elements at
time $t-1,t-2,\dots,t-\tau$, for some integers $\tau \geq 1$ and
$d \geq 1$. Stencil computations arise in many contexts, ranging from
iterative finite-difference methods for the numerical solution of
partial differential equations, to algorithms for the simulation of
cellular automata, as well as in dynamic programming algorithms and in
image-processing applications. Also, the simulation of a $d$-dimensional
mesh~\cite{BilardiP97} can be envisioned as a stencil computation.

In this subsection, we restrict our attention to stencil computations
with $\tau = 1$. To this purpose, we define the \emph{$(n,d)$-stencil}
problem which represents a wide class of stencil computations (see,
e.g.,~\cite{FrigoS05}). Specifically, the problem consists in evaluating
all nodes of a DAG of $n^{d+1}$ nodes, each represented by a distinct
tuple $\langle i_0,i_1,\ldots,i_d \rangle$, with $0 \leq i_0,i_1,\ldots,i_d
< n$, where each node $\langle i_0,i_1,\ldots,i_d \rangle$ is connected,
through an outgoing arc, to (at most) $3^d$ neighbors, namely $\langle
i_0+\delta_0,i_1+\delta_1,\ldots,i_{d-1}+\delta_{d-1},i_d+1\rangle$ for
each $\delta_0,\delta_1,\ldots,\delta_{d-1} \in \{0, \pm 1\}$ (whenever
such nodes exist). We suppose $n$ to be a power of two. Intuitively,
the $(n,d)$-stencil problem consists of $n$ time steps of a stencil
computation on a $d$-dimensional spatial grid of side $n$, where each
DAG node corresponds to a grid element (first $d$ coordinates) at a
given time step (coordinate $i_d$).

Let $\mathscr{C}_d$ denote the class of static algorithms for the
$(n,d)$-stencil problem such that any $\A \in \mathscr{C}_d$ for $q$
processing elements satisfies the following properties: (i) each DAG
node is evaluated once (i.e., recomputation is not allowed); (ii) no
processing element computes more than $\epsilon n^{d+1}$ DAG nodes,
for some constant $0 < \epsilon < 1$; (iii) all of the foldings of
$\A$ on $q'$ processing elements, $2 \leq q' < q$, also belong to
$\mathscr{C}_d$. Note that, as before, this class of algorithms makes 
no assumptions on the input and output distributions. The following
lemma establishes a lower bound on the communication complexity of
the algorithms in $\mathscr{C}_d$.

\begin{lemma}\label{Stencillowerbound}
The communication complexity of any $(n,d)$-stencil algorithm in
$\mathscr{C}_d$ when executed on \evald\ is $\BOM{n^d/p^{(d-1)/d}+\of}$.
\end{lemma}

\begin{proof}
The bound for $\of=0$ is proved in~\cite[Theorem~5]{ScquizzatoS14}, and it clearly
extends to the case $\of>0$. The additive $\of$ term follows since at least
one message is sent by some processing element.
\end{proof}

In what follows, we develop efficient network-oblivious algorithms for
the $(n,d)$-stencil problem, for the special cases of $d = \{1,2\}$.
The generalization to values $d > 2$, and to other types of stencils,
is left as an open problem.

\subsubsection{The $(n,1)$-Stencil Problem}
The $(n,1)$-stencil problem consists of the evaluation of a DAG shaped
as a $2$-dimensional array of side $n$.  We reduce the solution of the
stencil problem to the computation of a diamond DAG. Specifically, we
define a \emph{diamond DAG} of side $n$ as the intersection of a
$(2n-1,1)$-stencil DAG with the following four half-planes: $i_0 + i_1
\geq (n-1)$, $i_0 - i_1 \leq (n-1)$, $i_0 - i_1 \geq - (n-1)$, and
$i_0 + i_1 \leq 3(n-1)$ (i.e., the largest diamond included in the
stencil).\footnote{We observe that our definition of diamond DAG is
  consistent with the one in~\cite{BilardiP97}, whose edges are a
  superset of those of the diamond DAG defined
  in~\cite{AggarwalCS90}.} It follows that an $(n,1)$-stencil DAG can
be partitioned into five full or truncated diamond DAGs of side less
than $n$ which can be evaluated in a suitable order, with the outputs
of one DAG evaluation providing the inputs for subsequent DAG evaluations.

Our network-oblivious algorithm for the $(n,1)$-stencil is specified
on \spec{n} and consists of five \emph{stages}, where in each stage
the whole \spec{n} machine takes care of the evaluation of a distinct
diamond DAG (full or truncated) according to the aforementioned partition.
We require that all of the $\BO{n}$ inputs necessary for the evaluation
of a diamond DAG are evenly distributed among the $n$ VPs at the start
of the stage in charge of the DAG. No matter how the inputs are assigned
to the VPs at the beginning of the algorithm, the data movement required
to guarantee the correct input distribution at the various stages can be
accomplished in $\BO{1}$ 0-supersteps where each VP sends/receives
$\BO{n}$ messages. 

We now focus on the evaluation of the individual diamond DAGs. For ease
of presentation, we consider the evaluation of a full diamond DAG of
side $n$ on \spec{n}. Simple yet tedious modifications are required
for dealing with truncated or smaller diamond DAGs. We exploit the
fact that the DAG can be decomposed recursively into smaller diamonds.
(Oblivious parallel algorithms adopting this strategy have been studied
in~\cite{ChowdhuryR08,FrigoS09}, but these algorithms, specified for a
different computational framework, are not oblivious to the number of
processors $p$, hence they cannot be directly compared to ours.)

\begin{figure}[t]
    \begin{center}
       \includegraphics[width=0.6\textwidth]{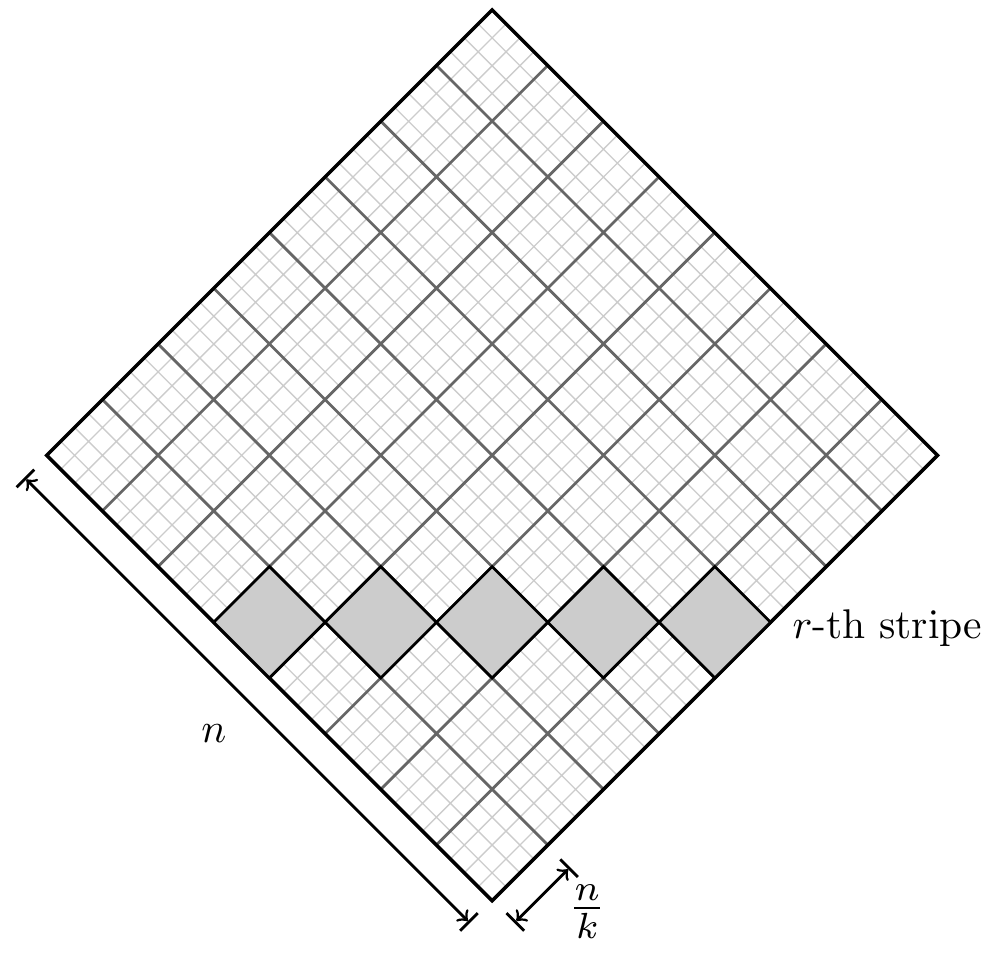}
       \caption{The decomposition of the diamond DAG performed by our algorithm.}
       \label{fig:diamonds}
    \end{center}
\end{figure}

Let $k = 2^{\lceil \sqrt{\log n} \rceil}$. The diamond DAG is
partitioned into $2k-1$ horizontal stripes, each containing up to $k$
diamonds of side $n/k$, as depicted in Figure~\ref{fig:diamonds}. The
DAG evaluation is accomplished into $2k-1$ non-overlapping \emph{phases}.
In the $r$-th such phase, with $0 \leq r < 2k-1$, the diamonds in the
$r$-th stripe are evaluated in parallel by distinct \spec{n/k} submachines
formed by disjoint segments of consecutively numbered VPs.\footnote{We
  observe that some \spec{n/k} submachines may not be assigned to
  subproblems, since the number of diamonds in a stripe could be
  smaller than $k$. In order to comply with the requirement that in
  the algorithm execution the sequence of superstep labels is the same
  at each processing element, we assume that idle \spec{n/k}
  submachines are assigned dummy diamonds of side $n/k$ to be
  evaluated.} At the beginning of each phase, a 0-superstep is
executed so to provide the VPs of each \spec{n/k} submachine with the
appropriate input, that is, the immediate predecessors (if any) of the
diamond assigned to the submachine. In this superstep each VP
sends/receives $\BO{1}$ messages. In each phase, the diamonds of side
$n/k$ are evaluated recursively.

In general, at the $i$-th recursive level, with $i \geq 1$, a total of
$(2k-1)^i$ non-overlapping phases are executed where diamonds of side
$n_i = n/k^i$ are evaluated in parallel by distinct \spec{n_i}
submachines. Each such phase starts with a superstep of label $(i-1)
\cdot \log k$ in order to provide each \spec{n_i} with the appropriate
input. In turn, the evaluation of a diamond of side $n_i$ within an
\spec{n_i} submachine is performed recursively by partitioning its
nodes into $2k-1$ horizontal stripes of diamonds of side $n_{i+1} =
n/k^{i+1}$ which are evaluated in $2k-1$ non-overlapping phases by
\spec{n_{i+1}} submachines, with each phase starting with a superstep
of label $i \cdot \log k$ where each VP sends/receives $\BO{1}$
messages (and thus each processor sends/receives $\BO{n/p}$ messages).
The recursion ends at level $\tau = \lfloor \log_k n \rfloor$,
which is the first level where the diamond of side $n_{\tau}$
becomes smaller than $k$. If $n_{\tau} > 1$, each diamond of side
$n_{\tau}$ assigned to an \spec{n_{\tau}} submachine is evaluated
straightforwardly in $2n_{\tau}-1$ supersteps of label $\tau \cdot
\log k$. Instead, if $n_{\tau} = 1$, at recursion level $\tau$ each VP
independently evaluates a 1-node diamond, and no communication is
required.

By unfolding the recursion, one can easily see that the evaluation of
a diamond DAG of side $n$ entails, overall, $(2k-1)^i$ supersteps of
label $(i-1) \cdot \log k$, for $1 \leq i \leq \tau$, and, if
$n_{\tau} > 1$, $(2k-1)^{\tau} n_{\tau}$ supersteps of label $\tau
\cdot \log k$. In each of these supersteps, every VP sends/receives
$\BO{1}$ messages. 

In order to guarantee $(\BT{1},n)$-wiseness of our algorithm, 
we assume that suitable dummy messages are added in the each superstep
so to make each VP exchange the same number of messages.

\begin{theorem}\label{thm:UB_1-stencil}
The communication complexity of the above network-oblivious algorithm
for the $(n,1)$-stencil problem when executed on \eval{p,\of} is
\[
H_{\textnormal{1-stencil}}(n,p,\of) = \BO{n4^{\sqrt{\log n}}},
\]
for every $1 < p \leq n$ and $0 \leq \of = \BO{n/p}$. The algorithm is
$(\BT{1},n)$-wise and $\BOM{1/4^{\sqrt{\log n}}}$-optimal with respect
to $\mathscr{C}_1$ on any \eval{p,\of} with $1 < p \leq n$ and $\of = \BO{n/p}$.
\end{theorem}

\begin{proof}
As observed before, the communication required at the beginning of
each of the five stages contributes an additive factor $\BO{n}$ to
the communication complexity, hence it is negligible. Let us then
concentrate on the communication complexity for one diamond DAG
evaluation. Recall that $\tau = \lfloor \log_k n \rfloor$. First
suppose that $p \leq k^\tau$. Observe that at every recursion level
$i$, with $0 \leq i < \lceil \log_k p \rceil$, the evaluation of each
diamond of side $n_i = n/k^i$ is performed by $p/k^i > 1$ processors,
while at every recursion level $i$, with $\lceil \log_k p \rceil \leq i
\leq \tau$, each diamond of side $n/k^i$ is evaluated by a single processor
of \eval{p,\of} and no communication takes place. Then, by the previous
discussion it follows that
\begin{eqnarray*}
H_{\textnormal{1-stencil}}(n,p,\of) & = & 
\BO{
\sum_{i=0}^{\lceil \log_k p\rceil-1} (2k-1)^{i+1}\left(\frac{n}{p} +
\of\right)} \\
& = & \BO{(2k)^{\log_k p+1}\frac{n}{p}} \\
& = & \BO{n2^{\log_k p}k} \\
& = & \BO{n4^{\sqrt{\log n}}}, 
\end{eqnarray*}
where we exploited the upper bound on $\of$. Instead, if $k^\tau < p \leq n$,
we have that at every recursion level $i$, with $0 \leq i \leq \tau$, the
evaluation of each diamond of side $n_i = n/k^i$ is performed by $p/k^i > 1$
processors. Then, by the above discussion and recalling that for $i = \tau$,
diamonds of side $n_\tau = n/k^\tau$ are evaluated straightforwardly
in $2n_\tau - 1$ supersteps, we obtain
\begin{eqnarray*}
H_{\textnormal{1-stencil}}(n,p,\of) & = &
\BO{\sum_{i=0}^{\tau-1} (2k-1)^{i+1}\left(\frac{n}{p} + \of\right)}
+
\BO{(2k-1)^{\tau}\frac{n}{k^{\tau}}\left(\frac{n}{p} + \of\right)} \\
& = & \BO{(2k)^{\tau}\frac{n}{k^{\tau}}\frac{n}{p}} \\
& = & \BO{n 2^{\tau} k} \\
& = & \BO{n4^{\sqrt{\log n}}},
\end{eqnarray*}
where we exploited the upper bound on $\of$ and the fact that $p >
k^{\tau}$ hence, by definition of $\tau$, $n/p < k$. The wiseness is
ensured by the dummy messages.  It is easy to see that the algorithm
complies with the requirements for belonging to $\mathscr{C}_1$, hence
the claimed optimality is a consequence of Lemma~\ref{Stencillowerbound},
and the theorem follows.
\end{proof}

Finally, we show that the network-oblivious algorithm for the $(n,1)$-stencil
problem achieves $\BOM{1/4^{\sqrt{\log n}}}$-optimality on the D-BSP as
well, for wide ranges of machine parameters.

\begin{corollary}
The above network-oblivious algorithm for the $(n,1)$-stencil problem 
is $\BOM{1/4^{\sqrt{\log n}}}$-optimal with respect to $\mathscr{C}_1$
on any \execd\ machine with $1 < p \leq n$, non-increasing $g_i$'s and
$\ell_i/g_i$'s, and $\ell_0/g_0 = \BO{n/p}$.
\end{corollary}

\begin{proof}
The corollary follows by Theorem~\ref{thm:UB_1-stencil} and by applying
Theorem~\ref{thm:optimality} with $p^\star = n$, $\of_i^m = 0$, and
$\of_i^M = \BT{n/2^i}$.
\end{proof}

We remark that a tighter analysis of the algorithm and/or the adoption
of different values for the recursion degree $k$, still independent of
$p$ and $\of$, may yield slightly better efficiency. However, it is
an open problem to devise a network-oblivious algorithm which is
$\BT{1}$-optimal on the D-BSP for wide ranges of the machine
parameters.

\subsubsection{The $(n,2)$-Stencil Problem}
In this subsection we present a network-oblivious algorithm for the
$(n,2)$-stencil problem, which requires the evaluation of a DAG shaped
as a 3-dimensional array of side $n$. Both the algorithm and its
analysis are a suitable adaptation of the ones for the $(n,1)$-stencil
problem. In order to evaluate a $3$-dimensional domain we make use of
two types of subdomains which, intuitively, play the same role as the
diamond for the $(n,1)$-stencil: the octahedron and the
tetrahedron. An octahedron of side $n$ is the intersection of a
$(2n-1,2)$-stencil with the following eight half-spaces: $i_0 + i_2
\geq (n-1)$, $i_0 - i_2 \leq (n-1)$, $i_0 - i_2 \geq - (n-1)$, $i_0 +
i_2 \leq 3(n-1)$, $i_0 + i_1 \geq (n-1)$, $i_0 - i_1 \leq (n-1)$, $i_0
- i_1 \geq - (n-1)$, and $i_0 + i_1 \leq 3(n-1)$; a tetrahedron of
side $n$ is the intersection of a $(2n-1,2)$-stencil with the
following four half-spaces: $i_0 + i_1 \geq (n-1)$, $i_0 - i_1 \geq
(n-1)$, $i_1 + i_2 \leq 2(n-1)$, and $i_1 - i_2 \leq 0$.

As shown in~\cite{BilardiP97}, a $3$-dimensional array of side $n$ can
be partitioned into 17 instances of (possibly truncated) octahedra or
tetrahedra of side $n$ (see Figure~6 of~\cite{BilardiP97}). Our
network-oblivious algorithm exploits this partition and is specified
on \spec{n^2}. It consists of 17 stages, where in each stage the VPs
take care of the evaluation of one polyhedra of the partition. We
assume that at the beginning of the algorithm the inputs are evenly
distributed among the $n^2$ VPs, and also impose that the inputs of
each stage be evenly distributed among the VPs. The data movement
required to guarantee the correct input distribution for each stage
can be accomplished in $\BO{1}$ 0-supersteps, where each VP
sends/receives $\BO{1}$ messages.

Let $k = 2^{\lceil\sqrt{\log n}\rceil}$. An octahedron of side $n$ can
be partitioned into octahedra and tetrahedra of side $n/k$ in $\log k$
steps, where the $i$-th such step, with $1 \leq i \leq \log k$,
refines a partition of the initial octahedron into octahedra or
tetrahedra of side $n/2^{i-1}$ by decomposing each of these polyhedra
into smaller ones of side $n/2^i$, according to the scheme depicted in
\cite[Figure~5]{BilardiP97}. The final partition is obtained at the
end of the $\log k$-th step.  The octahedra and tetrahedra of the
final partition can be grouped in horizontal stripes in such a way
that the polyhedra of each stripe can be evaluated in
parallel. Consider first the set of octahedra of the partition. It can
be seen that the projection of these octahedra on the
$(i_0,i_2)$-plane coincides with the decomposition of the diamond DAG
depicted in Figure~\ref{fig:diamonds}. As a consequence, we can
identify $2k-1$ horizontal stripes of octahedra, where each stripe
contains up to $k^2$ octahedra of side $n/k$. Moreover, the
interleaving of octahedra and tetrahedra in the basic decompositions
of \cite[Figure~5]{BilardiP97}, implies that there is a stripe of
tetrahedra between each pair of consecutive stripes of
octahedra. Hence, there are also $(2k-1)-1$ horizontal stripes of
tetrahedra, each containing up to $k^2$ tetrahedra of side $n/k$.
Overall, the octahedron of side $n$ is partitioned into $4k-3$
horizontal stripes of at most $k^2$ polyhedra of side $n/k$ each,
where stripes of octahedra are interleaved with stripes of tetrahedra.
With a similar argument one can derive a partition of a tetrahedron of
side $n$ into $2k-1 \leq 4k-3$ horizontal stripes of at most $k^2$
polyhedra of side $n/k$ each, where stripes of octahedra are interleaved
with stripes of tetrahedra. 

Once established the above preliminaries, the network-oblivious algorithm
to evaluate a $3$-dimensional array of side $n$ on \spec{n^2} follows closely
from the recursive strategy used for the $(n,1)$-stencil problem: the evaluation
of an octahedron is accomplished in $4k-3$ non-overlapping phases,
in each of which the polyhedra (either octahedra or tetrahedra) of side
$n/k$ in one horizontal stripe of the partition described above
are recursively evaluated in parallel by distinct \spec{n^2/k^2}
submachines formed by disjoint segments of consecutively numbered VPs;
a tetrahedron of side $n$ can be evaluated through a recursive
strategy similar to the one for the octahedron within the same
complexity bounds. As usual, we add to each superstep
$\BO{1}$ dummy messages per VP to guarantee $(\BT{1},n^2)$-wiseness.

\begin{theorem}\label{thm:UB_2-stencil}
The communication complexity of the above network-oblivious algorithm
for the $(n,2)$-stencil problem when executed on \eval{p,\of} is
\[
H_{\textnormal{2-stencil}}(n,p,\of) = \BO{\frac{n^2}{\sqrt{p}} 8^{\sqrt{\log n}}},
\]
for every $1 < p \leq n^2$ and $0 \leq \of = \BO{n^2/p}$. The
algorithm is $(\BT{1},n^2)$-wise and $\BOM{1/8^{\sqrt{\log n}}}$-optimal
with respect to $\mathscr{C}_2$ on any \eval{p,\of} with $1 < p \leq n^2$
and $\of = \BO{n^2/p}$.
\end{theorem}

\begin{proof}
Let $H_{\textnormal{octahedron}}(n,p,\of)$ be the communication
complexity required by the recursive strategy presented above for the
evaluation of an octahedron of side $n$, when executed on
\eval{p,\of}. The recursion depth of
that strategy is $\tau = \lfloor \log_k n \rfloor$. First
suppose that $p \leq k^{2\tau}$. At every recursion level
$i$, with $0 \leq i < \lceil (\log_k p)/2 \rceil$, the evaluation of
each polyhedron of side $n_i = n/k^i$ is performed by $p/k^{2i} > 1$
processors, while at every recursion level $i$, with $\lceil (\log_k
p)/2 \rceil \leq i \leq \tau$, each polyhedron of side $n/k^i$ is
evaluated by a single processor of \eval{p,\of} and no communication
takes place. Then, by the previous discussion it follows that
\begin{align*}
H_{\textnormal{octahedron}}(n,p,\of) 
&= \BO{\sum_{i=0}^{\lceil(\log_k p)/2\rceil - 1}(4k-3)^{i+1} \left(\frac{n^2}{p} + \of\right)} \\
&= \BO{(4k)^{(\log_k p)/2 + 1} \frac{n^2}{p}} \\
&= \BO{\frac{n^2}{\sqrt{p}} 2^{\log_k p} k} \\
&= \BO{\frac{n^2}{\sqrt{p}} 4^{\sqrt{\log n}}},
\end{align*}
where we used the hypothesis $\of = \BO{n^2/p}$. Instead, when
$k^{2\tau} < p \leq n^2$, we have that at every recursion level $i$,
with $0 \leq i \leq \tau$, the evaluation of each polyhedron of side
$n_i = n/k^i$ is performed by $p/k^{2i} > 1$ processors. Then, since
for $i = \tau$, polyhedra of side
$n_{\tau} = n/k^{\tau}$ are evaluated straightforwardly in
$\BT{n_{\tau}}$ supersteps, we obtain
\begin{align*}
H_{\textnormal{octahedron}}(n,p,\of)
&= \BO{\sum_{i=0}^{\tau-1}(4k-3)^{i+1}\left(\frac{n^2}{p} + \of\right)} +
\BO{(4k-3)^{\tau}\frac{n}{k^{\tau}}\left(\frac{n^2}{p} + \of\right)} \\
&= \BO{(4k)^{\tau} \frac{n}{k^{\tau}} \frac{n^2}{p}} \\
&= \BO{\frac{n^2}{\sqrt{p}} 4^{\tau} k} \\
&= \BO{\frac{n^2}{\sqrt{p}} 8^{\sqrt{\log n}}},
\end{align*}
where we used the hypothesis $\of = \BO{n^2/p}$ and the inequalities
$n/k^{\tau} < k$ and $k^{2\tau} < p$.
Similar upper bounds on the communication complexity can be proved for
the evaluation of a tetrahedron of side $n$ and for the evaluation of
truncated octahedra or tetrahedra.  

Recall that the algorithm for the
$(n,2)$-stencil problem consists of 17 stages, where in each stage the
VPs take care of the evaluation of one (possibly truncated) octahedron
or tetrahedron of side $n$, and that the data movement which
ensures  the correct input distribution for each stage can be
accomplished in $\BO{1}$ 0-supersteps, where each VP sends/receives
$\BO{1}$ messages. This implies that
\[
H_{\textnormal{2-stencil}}(n,p,\of)
=
\BO{\frac{n^2}{\sqrt{p}} 8^{\sqrt{\log n}}}.
\]
Since the strategies for the evaluation of (possibly truncated)
octahedra or tetrahedra can be made $(\BT{1},n^2)$-wise, through the
introduction of suitable dummy messages, the overall algorithm is also
$(\BT{1},n^2)$-wise. Moreover, the algorithm complies with the
requirements for belonging to $\mathscr{C}_2$, hence the claimed
optimality is a consequence of Lemma~\ref{Stencillowerbound}.
\end{proof}

\begin{corollary}
The above network-oblivious algorithm for the $(n,2)$-stencil problem is
$\BOM{1/8^{\sqrt{\log n}}}$-optimal with respect to $\mathscr{C}_2$ on
any \execd\ machine with $1 < p \leq n^2$, non-increasing $g_i$'s and
$\ell_i/g_i$'s, and $\ell_0/g_0 = \BO{n^2/p}$.
\end{corollary}

\begin{proof}
The corollary follows by Theorem~\ref{thm:UB_2-stencil} and by applying
Theorem~\ref{thm:optimality} with $p^\star = n^2$, $\of_i^m = 0$,
and $\of_i^M = \BT{n^2/2^i}$.
\end{proof}

\subsection{Limitations of the Oblivious Approach}\label{sec:alg_broadcast}
In this subsection, we establish a negative result by showing that for
the broadcast problem, defined below, a network-oblivious algorithm can
achieve $\BO{1}$-optimality on \eval{p,\of} only for very limited
ranges of $\of$. Let $V[0, \ldots, n-1]$ be a vector of
$n$ entries. The \emph{$n$-broadcast problem} requires to copy the
value $V[0]$ into all other $V[i]$'s.  Let $\mathscr{C}$
denote the class of static algorithms for the $n$-broadcast problem
such that any $\A \in \mathscr{C}$ for $q$ processing elements
satisfies the following properties: (i) 
at least $\epsilon q$ processing elements hold entries of $V$,
for some constant $0 < \epsilon \leq 1$, and the
distribution of the entries of $V$ among the processing elements
cannot change during the execution of the algorithm; (ii) all of the
foldings of $\A$ on $q'$ processing elements, $2 \leq q' < q$, also
belong to $\mathscr{C}$. The following lemma establishes a lower bound
on the communication complexity of the algorithms in $\mathscr{C}$

\begin{theorem}\label{thm:LBbroadcast}
The communication complexity of any $n$-broadcast algorithm in
$\mathscr{C}$ when executed on \evald, with $1 < p \leq n$ and
$\sigma \geq 0$, is
$\BOM{\max\{2,\of\} \log_{\max\{2,\of\}} p}$.
\end{theorem}

\begin{proof}
Let $\A$ be an algorithm in $\mathscr{C}$.  Suppose that the execution
of $\A$ on \evald\ requires $t$ supersteps, and let $p_i$ denote the
number of processors that ``know'' the value $V[0]$ by the end of the
$i$-th superstep, for $1 \leq i \leq t$. Clearly, $p_0 = 1$ and $p_t
\geq \epsilon p$, since by definition of $\mathscr{C}$, at least
$\epsilon p$ processors hold entries of $V$ to be updated with the
value $V[0]$. During the $i$-th superstep, $p_i-p_{i-1}$ new
processors get to know $V[0]$. Since at the beginning of this
superstep only $p_{i-1}$ processors know the value, we conclude that
the superstep involves an $h$-relation with $h \geq
\left\lceil (p_i-p_{i-1})/p_{i-1} \right\rceil$. 
Therefore, the communication complexity of $\A$ is
\begin{eqnarray*}
H_{\A}(n,p,\of) \geq 
\sum_{i=1}^{t}\left(\left\lceil \dfrac{p_i-p_{i-1}}{p_{i-1}}\right\rceil
+\of \right)
=
\sum_{i=1}^{t}\left(\left\lceil \dfrac{p_i}{p_{i-1}}\right\rceil
-1+\of \right).
\end{eqnarray*}
Assuming, without loss of generality, that the $p_i$'s are strictly
increasing, we obtain
\[
H_{\A}(n,p,\of)
= \BOM{t \max\{2,\of\} + \sum_{i=1}^{t} \dfrac{p_i}{p_{i-1}}}.
\]
Since $\prod_{i=1}^{t} {p_i/p_{i-1}} = p_t$, it follows that
$\sum_{i=1}^{t} p_i/p_{i-1}$ is minimized for $p_i/p_{i-1} = (p_t)^{1/t}
\geq (\epsilon p)^{1/t}$, for $1 \leq i \leq t$. Hence,
\begin{equation}\label{eq:genericLBbroadcast}
H_{\A}(n,p,\of) = \BOM{t \left(\max\{2,\of\}+ p^{1/t} \right)}.
\end{equation}
Standard calculus shows that the right-hand side is minimized (to
within a constant factor) by choosing 
$t = \BT{\log_{\max\{2,\of\}} p}$, and the claim follows.
\end{proof}

The above lower bound is tight. Consider the following
\evald\ algorithm for $n$-broadcast.  Let the entries of $V$ be evenly
distributed among the processors, with $V[0]$ held by processor
$\text{P}_0$.  For convenience we assume $n$ is a power of two.  Let
$\kappa$ be the smallest power of 2 greater than or equal to
$\max\{2,\of\}$.  The algorithm consists of $\lceil \log_{\kappa} p
\rceil$ supersteps: in the $i$-th superstep, with $0 \leq i < \lceil
\log_{\kappa} p\rceil$, each $\text{P}_{j p/\kappa^i}$, with $0 \leq j
< \kappa^i$, sends the value $V[0]$ to $\text{P}_{(\kappa j+\ell)p/
  \kappa^{i+1}}$, for each $0 \leq \ell < \kappa$. (When
$\log_{\kappa} p$ is not an integer value, in the last superstep only
values of $\ell$ that are multiples of $\kappa^{i+1}/p$ are used.)  It
is immediate to see that the algorithm belongs to $\mathscr{C}$ and
that its communication complexity on \evald\ is
\[
H_{\textnormal{broad}_\kappa}(n,p,\of) =
\BO{(\kappa + \of)\log_{\kappa} p}
= 
\BO{\max\{2,\of\} \log_{\max\{2,\of\}} p}.
\]
Therefore, the algorithm is $\BO{1}$-optimal. Observe that the algorithm 
is aware of parameter $\of$ and, in fact, this knowledge is
crucial to achieve optimality. In order to see this,
we prove that \emph{any} network-oblivious algorithm
for $n$-broadcast can be $\BT{1}$-optimal on \eval{p,\of}, only for
limited ranges of $\of$.  Let $H^{\star}(n,p,\of)$ denote the best
communication complexity achievable on \eval{p,\of} by an algorithm
for $n$-broadcast belonging to $\mathscr{C}$. By the above discussion
we know that $H^{\star}(n,p,\of) = \BT{\max\{2,\of\}
  \log_{\max\{2,\of\}} p}$.  Let $\A \in \mathscr{C}$ be a
network-oblivious algorithm for $n$-broadcast specified on
\spec{v(n)}.  For every $1 < p \leq v(n)$ and $0 \leq \of_1 \leq
\of_2$, we define the maximum slowdown incurred by $\A$ with respect
to the best \eval{p,\of}-algorithm in $\mathscr{C}$, for $\of \in
[\of_1,\of_2]$, as
\[
\GAP_{\A}(n,p,\of_1,\of_2) =
\max_{\substack{\of_1 \leq \of \leq \of_2}}
\left\{\dfrac{H_{\A}(n,p,\of)}{H^{\star}(n,p,\of)}\right\}.
\]
\begin{theorem}\label{thm:impossibility}
Let $\A \in \mathscr{C}$ be a network-oblivious algorithm for
$n$-broadcast specified on \spec{v(n)}.  For every $1 < p \leq v(n)$
and $0 \leq \of_1 \leq \of_2$, we have
\begin{equation*}
\GAP_{\A}(n,p,\of_1, \of_2) = 
\BOM{\dfrac{\log\max\{2,\of_2\}}
           {\log \max\{2,\of_1\} + \log\log \max\{2,\of_2\}}}.
\end{equation*}
\end{theorem}

\begin{proof}
The definition of function $\GAP$ implies that 
\[
\GAP_{\A}(n,p,\of_1,\of_2) = 
\BOM{
\dfrac{H_{\A}(n,p,\of_1)}{H^{\star}(n,p,\of_1)}
+
\dfrac{H_{\A}(n,p,\of_2)}{H^{\star}(n,p,\of_2)}
}.
\]
Let $t$ be the number of supersteps executed by the folding of 
$\A$ on \eval{p,\of}, and note that, since $\A$ is network-oblivious, 
this number cannot depend on $\of$.  By arguing as in the proof of
Theorem~\ref{thm:LBbroadcast} (see
Inequality~\ref{eq:genericLBbroadcast}) we get that 
$H_{\A}(n,p,\of) = \BOM{t \left(\max\{2,\of\}+ p^{1/t} \right)}$, 
for any $\of$, hence $\GAP_{\A}(n,p,\of_1,\of_2)$ is bounded from
below by
\[
\BOM{
\dfrac{t \left(\max\{2,\of_1\}+ p^{1/t} \right)}
{\max\{2,\of_1\}\log_{\max\{2,\of_1\}} p}
+
\dfrac{t \left(\max\{2,\of_2\}+ p^{1/t} \right)}
{\max\{2,\of_2\}\log_{\max\{2,\of_2\}} p}
},
\]
which is minimized for $t = \BT{\log p/(\log \max\{2,\of_1\} + \log
\log \max \{2,\of_2\})}$.  Substituting this value of $t$ in the above
formula yields the stated result.
\end{proof}

An immediate consequence of the above theorem is that if a
network-oblivious algorithm for $n$-broadcast is $\BT{1}$-optimal on
\eval{p,\of} it cannot be simultaneously $\BT{1}$-optimal on an
\eval{p,\of'}, for any $\of'$ sufficiently larger than $\of$.
A similar limitation of the optimality of a network-oblivious
algorithm for $n$-broadcast can be argued with respect to 
its execution on \execd.

\section{Extension to the Optimality Theorem}\label{sec:extension}
The optimality theorem of Section~\ref{sec:optimality} makes crucial
use of the wiseness property. Broadly speaking, a network-oblivious
algorithm is $(\BT{1},p)$-wise when the communication performed in the
various supersteps is somewhat balanced in the sense that the maximum
number of messages sent/received by a virtual processor does not
differ significantly from the average number of messages sent/received
by other virtual processors belonging to the same region of suitable
size.  While there exist $(\BT{1},p)$-wise network-oblivious
algorithms for a number of important problems, as shown in
Section~\ref{sec:algorithms}, there cases where wiseness may not be
guaranteed.

As a simple example of poor wiseness, consider a network-oblivious
algorithm $\A$ for \spec{n} consisting of one $0$-superstep where
$\text{VP}_0$ sends $n$ messages to $\text{VP}_{n/2}$.  Fix $p$ with
$2\leq p \leq n$. Clearly, for each $1 \leq j \leq \log p$ we have
that $H_{\A}(n,2^j, 0)=n$, hence the algorithm is $(\alpha, p)$-wise
only for $\alpha = \BO{1/p}$.  When executed on a \exec{p, \bfm{g},
  \mathbf{0}}, the communication time of the algorithm is $ng_0$.
However, as already observed in \cite{BilardiPP07}, under reasonable
assumptions the communication time of the algorithm's execution on
the D-BSP can be improved by first evenly spreading the $n$ messages
among clusters of increasingly larger size which include the sender,
and then gathering the messages within clusters of increasingly smaller
size which include the receiver. Motivated by this observation, we
introduce a more effective protocol to execute network-oblivious
algorithms on the D-BSP. By employing this protocol we are able to
prove an alternative optimality theorem which requires a much weaker
property than wiseness, at the expense of a slight (polylogarithmic)
loss of efficiency.

Let $\A$ be a network-oblivious algorithm specified on \spec{v(n)},
and consider its execution on a \exec{p, \bfm{g}, \bfm{\ell}}, with
$1 \leq p \leq v(n)$. As before, each D-BSP processor $\text{P}_j$,
with $0\leq j<p$, carries out the operations of the $v(n)/p$
consecutively numbered VPs of \spec{v(n)} starting with
$\text{VP}_{j(v(n)/p)}$.  However, the communication required by each
superstep is now performed on D-BSP more effectively by enforcing a
suitable balancing.  More precisely, each $i$-superstep $s$ of $\A$,
with $0 \leq i < \log p$, is executed on the D-BSP through the
following protocol, which we will refer to as \emph{novel protocol}:
\begin{enumerate}
\item
(\textit{Computation phase}) Each D-BSP processor 
performs the local computations of its assigned virtual processors.
\item
(\textit{Ascend phase}) For $k=\log p-1$ downto $i+1$: within each
$k$-cluster $\Gamma_k$, the messages  which originate in
$\Gamma_k$ but are destined outside $\Gamma_k$, are evenly distributed
among the $p/2^k$ processors of $\Gamma_k$.
\item
(\textit{Descend phase}) For $k=i$ to $\log p-1$: within each
  $k$-cluster $\Gamma_k$, the messages currently residing in
  $\Gamma_k$ are evenly distributed among the processors of the
$(k+1)$-clusters inside $\Gamma_k$ which contain their final 
destinations.
\end{enumerate}
Observe that each iteration of the ascend/descend
phases requires a prefix-like computation to assign suitable
intermediate destinations to the messages in order to guarantee their
even distribution in the appropriate clusters. 
\begin{lemma}\label{lem:sim_dbsp}
Let $\A$ be a network-oblivious algorithm specified on \spec{v(n)} and
consider its execution on \exec{p,\bfm{g},\bfm{\ell}}, with $1 < p
\leq v(n)$, using the novel protocol. Let $\xi_s$ be the sequence of
supersteps employed by the protocol for executing an $i$-superstep $s$ of
$\A$, where $0 \leq i < \log p$.  Then, for every $i< k < \log p$,
$\xi_s$ comprises $\BO{1}$ $k$-supersteps of degree $\BO{2^{k}
  h_{\A}^s(n,2^{k})/p}$, and $\BO{\log p}$ $k$-supersteps each of
constant degree.
\end{lemma}

\begin{proof}
Consider iteration $k$ of the ascend phase of the protocol, with
$i+1\leq k \leq \log p-1$, and a $k$-cluster $\Gamma_k$. As invariant
at the beginning of the iteration, we have that the at most
$h_{\A}^s(n,2^{k+1})$ messages originating in each $k+1$-cluster
$\Gamma'$ included in $\Gamma_k$ and destined outside $\Gamma_k$ are
evenly distributed among the processors of $\Gamma'$. Hence, the even
distribution of these messages among the $p/2^k$ processors of
$\Gamma_k$ requires a prefix-like computation and an $\BO{\lceil
  2^{k+1} h_{\A}^s(n,2^{k+1})/p\rceil}$-relation within
$\Gamma_k$. Consider now iteration $k$ of the descend phase of the
protocol, with $i\leq k \leq \log p-1$, and a $k$-cluster
${\Gamma_k}$. As invariant at the beginning of the iteration, we have
that the at most $2h_{\A}^s(n,2^{k+1})$ messages to be moved in the
iteration are evenly distributed among the processors of $\Gamma_k$.
Since each $(k+1)$-cluster included in $\Gamma_k$ receives at most
$h_{\A}^s(n,2^{k+1})$ messages, the iteration requires a prefix-like
computation and an $\BO{\lceil 2^{k+1}
  h_{\A}^s(n,2^{k+1})/p\rceil}$-relation within $\Gamma_k$.  The lemma
follows, since each prefix-like computation in a $k$-cluster can be
performed in $\BO{\log p}$ $k$-supersteps of constant degree (e.g.,
using a straightforward tree-based strategy~\cite{JaJa92}).
\end{proof}

We now define the notion of fullness which is weaker than 
wiseness but which still allows us to port the optimality of
network-oblivious algorithms with respect to the evaluation model onto
the execution machine model, although incurring some loss of
efficiency.
\begin{definition}\label{def:full}
A static network-oblivious algorithm $\A$ specified on \spec{v(n)} is
said to be \emph {$(\gamma,p)$-full}, for some $\gamma > 0$ and
$1 < p \leq v(n)$, if considering the folding of 
$\A$ on \eval{2^{j},0} we have 
\[
\sum_{i=0}^{j-1} F_{\A}^i\left(n, 2^j\right) \geq \gamma \frac{p}{2^j} \sum_{i=0}^{j-1}
S_{\A}^i(n),
\]
for every $1 \leq j \leq \log p$ and input size $n$.
\end{definition}

It is easy to see that a $(\BT{1},p)$-wise network-oblivious algorithm
$\A$ is also $(\BT{1},p)$-full as long as $h_{\A}^s(n,p) \geq 1$, for
every $i$-superstep $s$ of $\A$ and every $1 < p \leq v(n)$. On the
other hand, a $(\BT{1},p)$-full algorithm is not necessarily
$(\BT{1},p)$-wise, as witnessed by the previously mentioned
network-oblivious algorithm consisting of a single 0-superstep where
$\text{VP}_0$ sends $n$ messages to $\text{VP}_{n/2}$, which is
$(\BT{1},p)$-full but not $(\BT{1},p)$-wise, for any $2\leq p \leq n$.
In this sense, $(\gamma,p)$-fullness is a weaker condition than
$(\BT{1},p)$-wiseness.

The following theorem shows that when $(\gamma,p)$-full algorithms are
executed on the D-BSP using the novel protocol, optimality in the
evaluation model is preserved on the D-BSP within a polylogarithmic
factor. As in Section~\ref{sec:optimality}, let $\mathscr{C}$ denote a
given class of static algorithms solving a problem $\Pi$, with the
property that for any algorithm $\A\in \mathscr{C}$, all its foldings
on $p\geq 2$ processing elements also belong to $\mathscr{C}$.

\begin{theorem}\label{thm:extoptimality}
Let $\A\in \mathscr{C}$ be a $(\gamma,p^\star)$-full network-oblivious
algorithm for some $\gamma>0$ and a power of two $p^*$. Let also
$\{\of^m_0, \ldots, \of^m_{\log p^{\star}-1}\}$ and $\{\of^M_0,
\ldots, \of^M_{\log p^{\star}-1}\}$ be two vectors of non-negative
values, with 
  $\of^m_j \leq \of^M_j$, for every $0 \leq j < \log p^{\star}$.
If $\A$ is $\beta$-optimal on \eval{2^{j},\of}
w.r.t. $\mathscr{C}$, for $\of_{j-1}^m \leq \of \leq \of^M_{j-1}$
and $1 \leq j \leq \log p^\star$, then, for every power of two
$p \leq p^\star$, $\A$ is $\BT{\beta/((1+1/\gamma)\log^2 p)}$-optimal
on \execd\ w.r.t. $\mathscr{C}$ when executed with the above protocol,
as long as:
\begin{itemize}
 \item the execution of $\A$ on \execd\ using the above protocol is in $\mathscr{C}$;
 \item $g_{i} \geq g_{i+1} $ and $\ell_{i}/g_{i} \geq \ell_{i+1}/g_{i+1}$, for $0 \leq i < \log p-1$;
 \item $\max_{1 \leq k \leq \log p}\{\of^m_{k-1} 2^{k} /p\} \leq \ell_{i}/g_{i} \leq
\min_{1 \leq k \leq \log p}\{\of^M_{k-1} 2^{k} /p\}$, for $0 \leq i < \log p$.
\end{itemize}
\end{theorem}

\begin{proof}
Consider the execution of $\A$ on a \exec{p, \bfm{g}, \bfm{\ell}}
using the novel protocol.  Let $\tilde \A$ denote the actual sequence
of supersteps performed on the D-BSP in this execution of $\A$. Note
that once the D-BSP parameters are fixed, $\tilde \A$ can be regarded
as a network-oblivious algorithm specified on \spec{p}.  Clearly, any
optimality considerations on the communication time of the execution
of $\tilde \A$ (regarded as a network-oblivious algorithm) on \exec{p,
  \bfm{g}, \bfm{\ell}} using the standard protocol, will also apply to
the communication time of the execution of $\A$ on \exec{p, \bfm{g},
  \bfm{\ell}} using the novel protocol, being the two communication
times the same.

We will assess the degree of optimality of the communication time of
$\tilde \A$ by resorting to Theorem~\ref{thm:optimality}. This entails
analyzing the communication complexity of $\tilde \A$ on \eval{2^{j},
  \of}, for any $1 \leq j \leq \log p$, and determining its wiseness.
Focus on \eval{2^{j}, \of} for some $1 \leq j \leq \log p$, and
consider an arbitrary $i$-superstep $s$ of $\A$, for some $0\leq i<j$.
Let $\xi_s$ be the sequence of supersteps in $\tilde \A$ executed in
the ascend and descend phases associated with superstep $s$.  From
Lemma~\ref{lem:sim_dbsp}, we know that for every $i< k < \log p$,
$\xi_s$ comprises $\BO{1}$ $k$-supersteps of degree $\BO{2^{k}
  h_{\A}^s(n,2^{k})/p}$, and $\BO{\log p}$ $k$-supersteps each of
constant degree.  Now, in the execution on \eval{2^{j}, \of} a
$k$-superstep with $k \geq j$ becomes local to the processors and does
not contribute to the communication complexity. Since each processor
of \eval{2^{j}, \of} corresponds to $p/2^j$ processors of \eval{p},
the communication complexity on \eval{2^{j}, \of} contributed by the
sequence $\xi_s$ is
\begin{equation*}
\BO{\sum_{k=i+1}^{j-1} \left(
\frac{p}{2^j}\left(
\dfrac{2^{k}}{p}h_{\A}^s(n,2^{k})
+ \log p \right) + \of \log p \right)
}.
\end{equation*}
Therefore, since $h_{\A}^s(n,2^{k}) \leq 2^{j-k}h_{\A}^s(n,2^j)$, the
above summation is upper bounded by
\[
\BO{\sum_{k=i+1}^{j-1} \left(
h_{\A}^s(n, 2^{j}) + \frac{p \log p}{2^j} + \of
\log p}\right)=
\BO{
\left(h_{\A}^s(n, 2^{j}) + \frac{p}{2^j} + \of
\right)\log^2 p}.
\]
Recall that $L_{\A}^i(n)$ denotes the set of $i$-supersteps executed
by $\A$, and $S^i_{\A}(n) = |L_{\A}^i(n)|$. Thus, the communication
complexity of $\tilde \A$ on \eval{2^{j}, \of} can be written as
\begin{eqnarray*}
H_{\tilde \A}(n, 2^{j}, \of) & = &
\BO{\sum_{i=0}^{j-1} \sum_{s \in L_{\A}^i(n)}  
\left(h_{\A}^s(n, 2^{j}) + \frac{p}{2^j} + \of
\right)\log^2 p} \\
& = &
\BO{\log^2 p 
\left(\sum_{i=0}^{j-1} \sum_{s \in L_{\A}^i(n)}  
\left(h_{\A}^s(n, 2^{j}) + \of
\right) +
\sum_{i=0}^{j-1} \sum_{s \in L_{\A}^i(n)}  
\frac{p}{2^j}\right)
} \\
& = &
\BO{\log^2 p 
\left(H_{\A}(n, 2^{j}, \of) +
\frac{p}{2^j}  \sum_{i=0}^{j-1} 
S^i_{\A}(n) \right)
} \\
& = &
\BO{(1+1/\gamma) \log^2 p \cdot H_{\A}(n, 2^{j}, \of)},
\end{eqnarray*}
where the last inequality follows by the $(\gamma,p^*)$-fullness of $\A$. 

The above inequality shows that algorithm $\tilde \A$ is
$\beta/((1+1/\gamma) \log^2 p)$-optimal as a consequence of the
$\beta$-optimality of $\A$. Let us now assess the wiseness of $\tilde
\A$.  Consider again the sequence $\xi_s$ of supersteps of $\tilde \A$
associated with an arbitrary $i$-superstep $s$ of $\A$, for some $0\leq
i< \log p$. We know that for every $i< k < \log p$,
$\xi_s$ comprises $\BO{1}$ $k$-supersteps of degree $\BO{2^{k}
  h_{\A}^s(n,2^{k})/p}$, and $\BO{\log p}$ $k$-supersteps each of
constant degree. Moreover, we can assume that suitable dummy messages
are added so that in a $k$-superstep of degree $\BO{2^{k}
  h_{\A}^s(n,2^{k})/p}$ (resp., degree $\BO{1}$) all processors of a
$(k+1)$-cluster send $\BT{2^{k} h_{\A}^s(n,2^{k})/p}$ (resp.,
$\BT{1}$) messages to the sibling $(k+1)$-cluster included in the same
$k$-cluster. It is easy to see that the above considerations about the
optimality of $\tilde \A$ remain unchanged, while $\tilde \A$ becomes
$(\BT{1},p)$-wise.  Finally, we recall that $\tilde \A$ belongs to
class $\mathscr{C}$ by hypothesis, and this is so even forcing it into
being wise.  Therefore, by applying Theorem~\ref{thm:optimality} to
$\tilde \A$, we can conclude that $\tilde \A$, hence $\A$, is
$\BT{\beta/((1+1/\gamma)\log^2 p)}$-optimal on a \execd\ with parameters
satisfying the initial hypotheses.
\end{proof}

We conclude this section by observing that the relation stated by the above theorem between
optimality in the evaluation model and optimality in D-BSP can be tighten when the $g_{i}$ and
$\ell_{i}$ parameters of the D-BSP decrease geometrically.  In this case, it is known that a
prefix-like computation within a $k$-cluster, for $0 \leq k < \log p$, can be performed in
$\BO{g_k+\ell_k}$ communication time (e.g., see \cite[Proposition 2.2.2]{BilardiPP07}).
Then, by a similar argument used to prove Theorem~\ref{thm:extoptimality} it can be shown
that a $(\gamma,p)$-full algorithm $\A$ which is $\beta$-optimal in the evaluation model
becomes $\BT{\beta/((1+1/\gamma)\log p)}$-optimal when executed on the D-BSP, thus
reducing by a factor $\log p$ the gap between the two optimality factors.

\section{Conclusions}\label{sec:conclusions}

We introduced a framework to explore the design of network-oblivious
algorithms, that is, algorithms which run efficiently on machines with
different processing power and different bandwidth/latency
characteristics, without making explicit use of architectural
parameters for tuning performance.  In the framework, a
network-oblivious algorithm is written for $v(n)$ virtual processors
(specification model), where $n$ is the input size and $v(\cdot)$ a
suitable function. Then, the performance of the algorithm is analyzed
in a simple model (evaluation model) consisting of $p\leq v(n)$
processors and where the impact of the network topology on
communication costs is accounted for by a latency parameter
$\of$. Finally, the algorithm is executed on the D-BSP
model~\cite{delaTorreK96,BilardiPP07} (execution machine model), which
describes reasonably well the behavior of a large class of
point-to-point networks by capturing their hierarchical structures. A
D-BSP consists of $p\leq v(n)$ processors and its network topology is
described by the $\log p$-size vectors $\bfm{g}$ and $\bfm{\ell}$,
which account for bandwidth and latency costs within nested clusters,
respectively. We have shown that for static network-oblivious
algorithms, where the communication requirements depend only on the
input size and not on the specific input instance (e.g., algorithms
arising in DAG computations), the optimality on the evaluation model
for certain ranges of $p$ and $\of$ translates into
optimality on the D-BSP model for corresponding ranges of 
the model's parameters. This result justifies the introduction of the
evaluation model that allows for a simple analysis of
network-oblivious algorithms while effectively bridging the
performance analysis to D-BSP which more accurately models the
communication infrastructure of parallel platforms through a
logarithmic number of parameters.

We devised $\BT{1}$-optimal static network-oblivious algorithms for
prominent problems such as matrix multiplication, FFT, and sorting,
although in the case of sorting optimality is achieved only when the
available parallelism is polynomially sublinear in the input
size. Also, we devised suboptimal, yet efficient, network-oblivious
algorithms for stencil computations, and we explored limitations of
the oblivious approach by showing that for the broadcasting problem
optimality in D-BSP can be achieved by a network-oblivious algorithm
only for rather limited ranges of the parameters.  Similar negative
results were also proved in the realm of cache-oblivious algorithms
(e.g., see~\cite{BilardiP01,BrodalF03,Silvestri06,Silvestri08}). Despite
these limitations, the pursuit of oblivious algorithms appears worthwhile
even when the outcome is a proof that no such algorithm can be
$\BT{1}$-optimal on an ample class of target machines. Indeed, the
analysis behind such a result is likely to reveal what kind of
adaptivity to the target machine is necessary to obtain optimal
performance. 

The present work can be naturally extended in several directions, some
of which are briefly outlined below. First, it would be useful to
further assess the effectiveness of our framework by developing novel
efficient network-oblivious algorithms for prominent problems beyond
the ones of this paper.  Some progress in this direction has been done
in~\cite{ChowdhuryRSB13,DemmelEFKLSS13}.
For the problems considered here, in
particular sorting and stencil computations, it would be very interesting
to investigate the potentiality of the network-oblivious approach at a
fuller degree.  More generally, it would be interesting to develop
lower-bound techniques to limit the level of optimality that
network-oblivious algorithms can reach on certain classes of target
platforms.  Another challenging goal concerns the generalization of
the results of Theorems~\ref{thm:optimality} and
\ref{thm:extoptimality} to algorithms for a wider class of problems,
and the weakening of the assumptions (wiseness or fullness) required to
prove these results.  It would be also useful to identify other
classes of machines, for which network-oblivious algorithms can be
effective.  Another natural question to investigate is how to
determine an efficient virtual-to-physical processor mapping for a
network with arbitrary topology.
Again, it would be very interesting to generalize our work to apply
to computing scenarios, such as traditional time-shared systems as well
as emerging global computing environments, where the amount of resources
devoted to a specific application can itself vary dynamically over time,
in the same spirit as~\cite{BenderEFGJM14} generalized the cache-oblivious
framework to environments in which the amount of memory available to an
algorithm can fluctuate. Finally, a somewhat ultimate goal is
represented by the integration of cache- and network-obliviousness in
a unified framework for the development of \emph{machine}-oblivious
computations, as recently done in~\cite{ChowdhuryRSB13} for shared-memory
platforms.

\paragraph*{Acknowledgments.}
The authors would like to thank Vijaya Ramachandran for helpful discussions.

\bibliographystyle{plain}
\bibliography{biblio}

\end{document}